\documentclass[a4paper,abstracton]{scrartcl}

\title{Approximation Algorithms for Generalized and Variable-Sized Bin Covering}
\author{Matthias Hellwig\thanks{Humboldt University of Berlin, Germany, \texttt{mhellwig@informatik.hu-berlin.de}} ~ and Alexander Souza\thanks{University of Freiburg, Germany, \texttt{souza@informatik.uni-freiburg.de}}}
\date{}

\usepackage{amsthm, amsmath,amsfonts, amssymb,savesym}
\usepackage{MnSymbol, enumerate}
\usepackage{ifthen}
\usepackage[utf8]{inputenc}
\usepackage{vmargin, cite}
\setpapersize{USletter}
\setmarginsrb{1.2in}{1.2in}{1.2in}{1.2in}{0mm}{0mm}{5mm}{5mm}

\newcommand{\ALG}[1][]{{\text{{\sc \ifthenelse{\equal{#1}{}}{alg}{#1}}}}}

\newcommand{\OPT}{\text{{\sc opt}}}
\newcommand{\opt}{\OPT}
\newcommand{\alg}[1][]{{\ALG{#1}}}

\newcommand{\R}{\mathbb{R}}

\newtheorem{theorem}{Theorem}

\newtheorem{observation}[theorem]{Observation}

\newtheorem{lemma}[theorem]{Lemma}

\renewcommand{\P}{{\cal P}}

\newcommand{\eps}{\varepsilon}

\newcommand{\command}[1]{{\bf #1}}

\newcommand{\BEGIN}[1][]{\command{begin} \ifthenelse{\equal{#1}{}}{}{\comment{ #1}} \\  }

\newcommand{\END}[1][]{\command{end} \ifthenelse{\equal{#1}{}}{}{\comment{ #1}} \\}

\newcommand{\comment}[1]{{\footnotesize // #1}}

\newcommand{\wlg}{w.\,l.\,o.\,g.}

\newcommand{\wlgs}{\wlg\ }
\savesymbol{div}
\DeclareMathOperator{\div}{div}
\renewenvironment{proof}[1][]{\noindent {\bf Proof\ifthenelse{\equal{#1}{}}{}{~(#1)}.}\ }{ \qed \\}
\newenvironment{proof2}[1][]{\noindent {\bf Proof\ifthenelse{\equal{#1}{}}{}{~(#1)}.}\ }{ \qed}

\usepackage{journals}

\newcommand{\NFD}{\ALG[nfd]}
\newcommand{\NFDS}{\NFD\ }
\newcommand{\NFDs}{\NFDS}
\usepackage{ifthen,graphicx}

\newtheorem{definition}[theorem]{Definition}
\newtheorem{example}[theorem]{Example}
\newtheorem{prop}[theorem]{Property}
\newcommand{\NP}{\text{\small NP}}
\renewcommand{\P}{{\text{\small P}}}
\newcommand{\prob}[1]{{\sc #1}}
\newcommand{\bigO}[1]{\text{O}(#1)}
\newboolean{full}
\setboolean{full}{true}

\begin{document}

\maketitle

\begin{abstract}
\noindent In this paper, we consider the \prob{Generalized Bin Covering} problem: We are given $m$ bin types, where each bin of type $i$ has profit $p_i$ and demand $d_i$. Furthermore, there are $n$ items, where item $j$ has size $s_j$. A bin of type $i$ is covered if the set of items assigned to it has total size at least the demand $d_i$. In that case, the profit of $p_i$ is earned and the objective is to maximize the total profit. To the best of our knowledge, only the cases $p_i = d_i = 1$ (\prob{Bin Covering}) and $p_i = d_i$ (\prob{Variable-Sized Bin Covering}) have been treated before. We study two models of bin supply: In the unit supply model, we have exactly one bin of each type, i.\,e., we have individual bins. By contrast, in the infinite supply model, we have arbitrarily many bins of each type. Clearly, the unit supply model is a generalization of the infinite supply model, since we can simulate the latter with the former by introducing sufficiently many copies of each bin. To the best of our knowledge the unit supply model has not been studied yet. It is well-known that the problem in the infinite supply model is $\NP$-hard, which can be seen by a straightforward reduction from \prob{Partition}, and this hardness carries over to the unit supply model. This also implies that the problem can not be approximated better than two, unless $\P = \NP$.

We begin our study with the unit supply model. Our results therein hold not only asymptotically, but for all instances. This contrasts most of the previous work on \prob{Bin Covering}, which has been asymptotic. We prove that there is a combinatorial $5$-approximation algorithm for \prob{Generalized Bin Covering} with unit supply, which has running time $\bigO{nm\sqrt{m+n}}$. This also transfers to the infinite supply model by the above argument. Furthermore, for \prob{Variable-Sized Bin Covering}, in which we have $p_i = d_i$, we show that the natural and fast \prob{Next Fit Decreasing} ($\NFD$) algorithm is a $9/4$-approximation in the unit supply model. The bound is tight for the algorithm and close to being best-possible, since the problem is inapproximable up to a factor of two, unless $\P = \NP$. Our analysis gives detailed insight into the limited extent to which the optimum can significantly outperform $\NFD$.

Then the question arises if we can improve on those results in asymptotic notions, where the optimal profit diverges. We discuss the difficulty of defining asymptotics in the unit supply model. For two natural definitions, the negative result holds that \prob{Variable-Sized Bin Covering} in the \emph{unit} supply model does not allow an APTAS. Clearly, this also excludes an APTAS for \prob{Generalized Bin Covering} in that model. Nonetheless, we show that there is an AFPTAS for \prob{Variable-Sized Bin Covering} in the \emph{infinite} supply model.
\end{abstract}

\section{Introduction}
\paragraph{Models and Motivation}
In this paper, we study generalizations of the \NP-hard classical \prob{Bin Covering} problem. In this problem, we have an infinite supply of unit-sized bins and a collection of items having individual sizes. The objective is to pack items into as \emph{many} bins as possible. That is, we seek to maximize the number of \emph{covered} bins, where a bin is covered if the total size of the packed items is at least the size of the bin. This problem is the dual of the classical \prob{Bin Packing} problem, where the goal is to pack the items into as \emph{few} bins as possible; see the survey~\cite{CoffmanGareyJohnson:1997}. \prob{Bin Covering} has received considerable attention in the past~\cite{AssmannJohnsonKleitmanLeung:1984,CsirikFrenkLabbeZhang:1999,CsirikFrenk:1990,CsirikTotik:1988,CsirikJohnsonKenyon:2001,JansenSolis-oba:2003}. We will survey relevant literature below.

In \prob{Generalized Bin Covering}, we have a set $I = \{1, \dots, m\}$ of \emph{bin types} and each bin $i \in I$ of some type has a \emph{profit} $p_i$ and \emph{demand} $d_i$. We denote the set of items by $J = \{1, \dots, n\}$ and define that each item $j \in J$ has a \emph{size} $s_j$. A bin is \emph{covered} or \emph{filled} if the total size of the packed items is at least the demand $d_i$ of the bin, in which case we earn profit $p_i$. The goal is to maximize the total profit gained. The special case with $p_i = d_i$ is known as \prob{Variable-Sized Bin Covering}. To the best of our knowledge, the model with general profits and demands has not been studied in the \prob{Bin Covering} setting before. Furthermore, we consider two models regarding the supply of bins: In the \emph{infinite supply model} -- as the name suggests -- we have arbitrary many bins available of each bin type. By contrast, we introduce the \emph{unit supply model}, in which we have one bin per type available, i.\,e., we speak of individual bins rather than bin types. Observe that the unit supply model is more general than the infinite supply model: By introducing $n$ copies of each bin, we can simulate the infinite supply model with the unit supply model. The converse is obviously not true. 

For motivating these generalizations, we mention the following two applications from trucking and canning. In the first application, suppose that a moving company receives a collection of inquiries for moving contracts. Each inquiry has a certain volume and yields a certain profit if it is served (entirely). The company has a fleet of trucks, where each truck has a certain capacity. The objective is to decide which inquiries to serve with the available trucks as to maximize total profit. This problem clearly maps to \prob{Generalized Bin Covering} in the unit supply model: the inquiries relate to the bins, while the trucks relate to the items. Notice that the unit supply model is essential here, since there the inquiries are individual, i.\,e., are not available arbitrarily often. Notice that all previous work on \prob{Bin Covering} exclusively considers the infinite supply model and is hence not applicable here. Also notice that the \prob{Generalized Bin Covering} problem applies especially if the profits do not necessarily correlate with the volume, but also depend on the types of goods. For example, shipping a smaller amount of valuables may yield higher profit than shipping a larger amount of books. As a second application, consider a canning factory, in which objects, e.g., fish, have to be packed into bins (of certain types having different sizes), such that the total packed weight reaches at least a certain respective threshold value. Here it is reasonable to assume that the available number of bins is arbitrary, i.\,e., the infinite supply model is suitable. If the profits are proportional to the threshold values, then we have an application for the \prob{Variable-Sized Bin Covering} problem.

Let $\mathcal{I}$ denote the family of all bin type sets and $\mathcal{J}$ the family of all item sets. Furthermore, let $\ALG(I, J)$ and $\OPT(I, J)$ be the respective profits gained by some algorithm $\ALG$ and by an optimal algorithm $\OPT$ on an instance $(I, J) \in \mathcal{I} \times \mathcal{J}$. The \emph{approximation ratio} of an algorithm $\ALG$, is defined by $\rho(\ALG) = \sup\{ \OPT(I, J)/\ALG(I, J) \mid I \in \mathcal{I}, J \in \mathcal{J} \}$. If $\rho(\ALG) \le \rho$ holds for an algorithm \prob{alg} with running time polynomial in the input size, then it is called a \emph{$\rho$-approximation}. If there is a $(1 + \eps)$-approximation for every $\eps > 0$, then the respective family of algorithms is called a \emph{polynomial time approximation scheme} (PTAS). If the running of a PTAS is additionally polynomial in $1/\eps$, then it is called a \emph{fully polynomial time approximation scheme} (FPTAS). With $\bar{\rho}(\ALG) = \lim_{p \rightarrow \infty} \sup\{ \OPT(I, J)/\ALG(I, J) \mid I \in \mathcal{I}, J \in \mathcal{J}, \OPT(I, J) \ge p \}$ we denote the \emph{asymptotic approximation ratio} of an algorithm $\ALG$. The notions of an asymptotic approximation algorithm and of asymptotic (F)PTAS (A(F)PTAS) transfer analogously.

\paragraph{Our Contribution}
In terms of results, we make the following contributions. In Section~\ref{sec:generalized_unit_supply} we consider \prob{Generalized Bin Covering} in the unit supply model. Our first main result is a $5$-approximation algorithm with running time $\bigO{nm\sqrt{m+n}}$ in Theorem~\ref{thm:gbc5approx}. The basic idea is to define an algorithm for a modified version of the problem. Even though this solution may not be feasible for the original problem, it will enable us to provide a good solution for the original problem. As a side result, which might be interesting in its own right, we obtain an integrality gap of two for a linear program of the modified problem and a corresponding integer linear program.

For \prob{Variable-Sized Bin Covering} in the infinite supply model, it is not hard to see that any reasonable algorithm (using only the largest bin type) is an asymptotic $2$-approximation. The situation changes considerably in the \emph{unit} supply model: Firstly, limitations in bin availability have to be respected. Secondly, the desired approximation guarantees are non-asymptotic (where we explain the issue concerning the asymptotics before Theorem~\ref{thm:lower_bound}). Our main result here is a tight analysis of the \prob{Next Fit Decreasing} ($\NFD$) algorithm in the unit supply model for \prob{Variable-Sized Bin Covering}, which can be found in Section~\ref{sec:variable-sized_unit_supply}. Theorem~\ref{thm:NFDisSuperb} states that $\NFD$ yields an approximation ratio of at most $9/4 = 2.25$ with running time $\bigO{n \log n + m \log m}$. The approximation guarantee is tight for the algorithm, see Example~\ref{exa:nfd_lower_bound}. The main idea behind our analysis is to classify bins according to their coverage: The bins that $\NFD$ covers with single items are -- in some sense -- optimally covered. If a bin is covered with at least two items, then their total size is at most twice the demand of the covered bin. Hence those bins yield at least half the achievable profit. Intuitively, the problematic bins are those that are not covered by $\NFD$: An optimal algorithm might recombine leftover items of $\NFD$ with other items to cover some of these bins and increase the profit gained. Our analysis gives insight into the limited extend to which such recombinations can be profitable. Firstly, our result is interesting in its own right, since $\NFD$ is a natural and fast algorithm. Secondly, it is also close to being best possible, in the following sense. 

A folklore reduction from \prob{Partition} yields that even the classical \prob{Bin Covering} problem is not approximable within a factor of two, unless $\P = \NP$. This clearly excludes the possibility of a PTAS for \prob{Bin Covering} in any of the models. The reduction crucially uses that there are only two identical bins in the \prob{Bin Covering} instance it creates. Then the question arises if one can improve in an asymptotic notion, where the optimal profit diverges. Indeed, for the classical \prob{Bin Covering} problem with \emph{infinite supply}, there actually is an A(F)PTAS~\cite{CsirikJohnsonKenyon:2001,JansenSolis-oba:2003}. 

However, since we have individual bins rather than bin types in the \emph{unit supply} model, there are difficulties for defining a meaningful asymptotics for \prob{Variable-Sized Bin Covering} therein. We discuss this issue in more detail before Theorem~\ref{thm:lower_bound}. Moreover, in Theorem~\ref{thm:lower_bound} we show that, even if there are $m > 2$ bins available and the optimal profit diverges, there are instances, for which no algorithm can have an approximation ratio smaller than $2 - \eps$ for an asymptotically vanishing $\eps > 0$, unless $\P = \NP$. Intuitively, we show that, even in this asymptotic notion, one still has to solve a \prob{Partition} instance on two ``large'' bins. Hence, for this asymptotics, there is no APTAS for \prob{Variable-Sized Bin Covering} in the \emph{unit} supply model, unless $\P = \NP$. However, this fact does \emph{not} exclude the possibility of an A(F)PTAS for \prob{Variable-Sized Bin Covering} in the \emph{infinite} supply model. Indeed, we can give an A(F)PTAS for \prob{Variable-Sized Bin Covering} with infinite supply. Our algorithm is an extension of the APTAS of Csirik et al.~\cite{CsirikJohnsonKenyon:2001} for classical \prob{Bin Covering}. We remove bin types with small demands and adjust the LP formulation and the rounding scheme used by~\cite{CsirikJohnsonKenyon:2001}. The running-time of the APTAS can then be further improved using the involved method of Jansen and Solis-Oba~\cite{JansenSolis-oba:2003} to yield the claimed AFPTAS in Theorem~\ref{thm:afptas}.

\paragraph{Related Work}

As already mentioned, to the best of our knowledge, all of the previous work considers the \prob{(Variable-Sized) Bin Covering} problem in the infinite supply model. Surveys on offline and online versions of these problems are given by Csirik and Frenk~\cite{CsirikFrenk:1990} and by Csirik and Woeginger~\cite{CsirikWoeginger:1998}.  Historically, research (on the offline version) of the \prob{Bin Covering} problem was initiated by Assmann et al.~\cite{AssmannJohnsonKleitmanLeung:1984}. They proved that \prob{Next Fit} is a $2$-approximation algorithm. Furthermore, they proved that \prob{First Fit Decreasing} is an asymptotic $3/2$-approximation and even improved on this result by giving an asymptotic $4/3$-approximate algorithm. Csirik et al.~\cite{CsirikFrenkLabbeZhang:1999} also obtained asymptotic approximation guarantees of $3/2$ and $4/3$ with simpler heuristic algorithms. 

The next breakthrough was made by Csirik, Johnson, and Kenyon~\cite{CsirikJohnsonKenyon:2001} by giving an APTAS for the classical \prob{Bin Covering} problem. The algorithm is based on a suitable LP relaxation and a rounding scheme. Later on, Jansen and Solis-Oba~\cite{JansenSolis-oba:2003} improved upon the running time and gave an AFPTAS. They reduce the number of variables by approximating the LP formulation of Csirik et al.~\cite{CsirikJohnsonKenyon:2001}, which yields the desired speed-up. Csirik and Totik~\cite{CsirikTotik:1988} gave a lower bound of $2$ for every online algorithm for online \prob{(Variable-Sized) Bin Covering}, i.\,e., items arrive one-by-one. This bound holds also asymptotically. 
Moreover, for \prob{Variable-Sized Bin Covering} any online algorithm must have an unbounded approximation guarantee, which can be seen from the following easy construction: There are two bins with demands $d_1 = n$ and $d_2 = 1$. The first arriving item has size $1$. If an online algorithm assigns this to bin $1$, no further items arrive. Otherwise, i.\,e., the item is assigned to bin $2$, an item of size $n-1$ arrives. This item can also only be assigned to bin $2$ and hence an algorithm gains profit at most $2d_2 =2$. It follows the competitive ratio is at the best equal to $n/2$. Thus this online model is only interesting from an asymptotic perspective. For online \prob{Variable-Sized Bin Covering} it is easy to see that the algorithm \prob{Next Fit} which uses only the largest bin type is already an asymptotic $2$-approximation. This in combination with the bound of Csirik and Totik~\cite{CsirikTotik:1988} already settles the online case. By contrast, there could be reasonable approximation guarantees in the offline model, both, non-asymptotically and with unit supply. However, to the best of our knowledge, no non-asymptotic offline version of \prob{Variable-Sized} or \prob{Generalized Bin Covering} has been considered previously.

\paragraph{Notation}

For any set $K \subseteq J$ define the \emph{total size} by $s(K) = \sum_{k \in K} s_k$. Note that a bin $i \in I$ is covered by a set $K \subseteq J$, if $s(K) \ge d_i$. As a shorthand, define $s = s(J)$. Any assignment of items to bins is a solution of the \prob{Generalized Bin Covering} problem. We will denote such an assignment by a collection of sets $S = (S_i)_{i \in I}$, where the $S_i \subseteq J$ are pairwise disjoint subsets of the set $J$ of items. Denote the profit of a solution $S$ by $p(S) = \sum_{i \in I:s(S_i)\geq d_i} p_i$. The profit of a solution $S$ determined by some algorithm \alg\ on an instance $(I, J)$ is denoted by $\ALG(I, J) = p(S)$. We may omit the instance $(I,J)$ in calculations, if it is clear to which instance $\ALG$ refers to. Furthermore, for a solution $S$ of an algorithm \alg, let $u_{\ALG}(i) = s(S_i)$ be the total size of the items assigned to bin $i$. If no confusion arises, we will write $u(i)$ instead of $u_{\ALG}(i)$.

\section{Generalized Bin Covering}
\label{sec:generalized_unit_supply}


\begin{theorem}\label{thm:gbc5approx}
There exists a $5$-approximation for \prob{Generalized Bin Covering} in the unit supply model, which has running time $\bigO{nm\sqrt{m+n}}$.
\end{theorem}

In terms of lower bounds, recall that the problem is inapproximable up to a factor of two, unless $\P=\NP$. In terms of upper bounds it is not hard to see that naive greedy strategies as that assign items to most profitable bins or that assign items to bins with the best ratio of profit to demand do not yield a constant approximation ratio. We firstly give an informal description of the ideas of our algorithm and define terms below. 

At the heart of our analysis for the upper bound lies the following observation. An optimal algorithm either covers a not too small fraction of bins with only one item exceeding the demand of the respective bin or a large fraction of bins is covered with more than one item, and all these items are smaller than the demand of the bin they were assigned to. We explain below why this can assumed to be true. In the former case we speak of singular coverage and in the latter of regular coverage. It is easy to see (cf. Observation~\ref{obs:matchingSingular}) that a bipartite maximum matching gives a solution being at least as good as the partial optimal solution of singularly covered bins. The more difficult case we have to handle is when a large fraction of bins is covered regularly in an optimal solution. We manage this problem by considering an appropriate modified \prob{Bin Covering} problem. In this problem items are only allowed to be assigned to bins with demand of at most their size. In this situation we say that the items are admissible to the respective bins. Further we are allowed to split items into parts and these parts may be distributed among the bins to which the whole item is admissible. Intuitively, in this modified problem the profit gained for a bin is the fraction of demand covered multiplied with the profit of the respective bin. 

In Lemma~\ref{lem:relaxedProblemOptimal} we show that the modified problem can be solved optimally in polynomial time by algorithm $\ALG^*$ defined in Figure~\ref{alg:algStar}. Algorithm $\ALG^*$ considers bins in non-increasing order of efficiency, where the efficiency of a bin is defined as the ratio of profit to demand of the respective bin. For each bin $i$ $\alg^*$ considers the largest item $j$, which is admissible to $i$. If $j$ was not assigned or only a part of $j$ was assigned previously then $j$ respectively the remaining part of $j$ is assigned to $i$. Then $\alg^*$ proceeds with the next smaller item. Once a bin is covered, the item which exceeds the bin is split so that the bin is exactly covered. Note that it can happen that during this procedure bins are assigned items, but are not covered. But due to the definition of the modified problem, these bins proportionally contribute to the objective function. A solution found by this algorithm is optimal, which we show by transforming an optimal solution to a linear program formulation of this modified problem into the solution of $\ALG^*$ without losing any profit.

A solution of $\alg^*$ can be transformed via two steps into a good solution for the \prob{Generalized Bin Covering} problem. By the way $\alg^*$ splits items we are able to reassemble the split items in Lemma~\ref{lem:transformSolution1} without losing too much profit in the modified model. The solution is further modified in a greedy way such that there are no items on a not covered bin $i$, which are admissible to another not covered bin $i'$ with larger efficiency. A solution with this property is called maximal with respect to the modified \prob{Bin Covering} problem. 

From a maximal solution we can create a solution for the \prob{Generalized Bin Covering} problem in Lemma~\ref{lem:transformSolution2}, again by losing only a bounded amount of profit. For this we move items successively from a not covered bin to the next not covered bin, which has at least the same efficiency. Since the solution was maximal the bins with higher efficiency are covered. By this procedure all bins are covered, which were not covered in the maximal solution, except the least efficient one. Either this least efficient bin or the remaining ones have at least half of the profit of all bins, which were not covered in the maximal solution. Therefore, after applying this procedure at most half of the profit is lost in comparison to the maximal solution. But now, all bins that receive items after this procedure are actually covered.

We start with the definitions we need in order to prove of Theorem~\ref{thm:gbc5approx}. Let $S=(S_1,\dots,S_m)$ be a solution. During the analysis we can assume that there are no unassigned items, i.\,e. all considered algorithms can assign all items, which is formally $S_1\cup \dots \cup S_m = J$. This is justified since we could add a dummy bin $m+1$ with $p_{m+1}=0$ and $d_{m+1}=\infty$ for sake of analysis.

A covered bin $i$ is called to be covered {\em singularly} if $S_i=\{j\}$ for some $j\in J$ with $s_j > d_i$, otherwise it is called to be covered {\em regularly}. Since we can assume that all items can be assigned to bins, we can also make the following assumptions. A bin $i$ which contains an item $j$ with $s_j > d_i$ is singularly covered. For a bin $i$ which is covered regularly it holds $s(S_i) \leq 2d_i$. The latter can be assumed to be true, since the bin $i$ does not contain an item $j$ with $s_j > d_i$ and hence in case $s(S_i) > 2d_i$ we could remove an item and the bin $i$ still would be covered.

\begin{observation}
Let $(I,J)$ be an instance and fix an optimal solution $O$ on this instance. Let $I_S$ the bins covered singularly in $O$ and $J_S$ the set of items on the bins $I_S$. There is an algorithm \prob{alg} such that $\ALG(I, J) \geq \OPT(I_S, J_S)$. The running time is $\bigO{nm\sqrt{m+n}}$.\label{obs:matchingSingular}
\end{observation}
\begin{proof}
We define an algorithm \prob{alg}, which simply solves the following instance for \prob{Maximum Weight Bipartite Matching} optimally: Define a bipartite graph $G = (I \cup J, E)$ with edges $E = \{ij \mid s_j > d_i \}$ and a weight function  $w : E \rightarrow \R$ given by $w_{ij} = p_i$ for $ij \in E$. Our algorithm \prob{alg} determines a \prob{Maximum Weight Bipartite Matching} $M \subseteq E$. Since our graph has $m+n$ nodes and at most $mn$ edges the algorithm of Hopcroft and Karp~\cite{HopcroftKarp:1973} gives a solution in time $\bigO{nm\sqrt{m+n}}$. The induced solution $S = (S_1,\dots,S_m)$ is $S_i = \{ j \}$ if $ij \in M$ and $S_i = \emptyset$ otherwise. 

Clearly, the singularly covered bins $I_S$ and the items $J_S$ assigned to them correspond to a matching in $G$. Thus $\ALG(I,J) \geq \ALG(I_S,J_S) = \OPT(I_S,J_S)$ by the correctness of the matching algorithm.
\end{proof}

Consider the following modified \prob{Bin Covering} problem. Items may be split into $p\geq 1$ parts. Then we consider an item $j$ as $p$ items $(j,1), \dots (j,p)$ of positive size, where we may omit the braces in indices. Denote $s_{j,i}$ the size of item part $(j,i)$ of item $j$. Formally it has to hold $s_j= \sum_{i=1}^p s_{j,i}$ and $s_{j,l} > 0$ for $1\leq l \leq p$. We refer to the $(j,l)$ as the parts of the item $j$. 

An item $j$ is said to be admissible to a bin $i$, if $s_j\leq d_i$. The parts $(j,l)$ of an item $j$ are defined to be admissible to $i$ if and only if $j$ is admissible to $i$.  Item parts can only be assigned to bins to which they are admissible.

For a fixed solution $S=(S_1,\dots,S_m)$ let $S_i$ be the set of item parts, assigned to bin $i$. Let $y_i := \min\{ s(S_i)/d_i, 1\}$ be the fill level of bin $i$ (note that the fill level of bin $i$ may be at most one, but nevertheless $s(S_i) > d_i$ is permitted, i.\,e. the sum of item sizes assigned to bin $i$ may exceed its demand).  The profit gained for bin $i$ in the modified problem is $p^*(S_i) := p_iy_i$, which is intuitively the percentage of covered demand multiplied with the profit of the bin, where the maximal profit which can be gained is bounded by $p_i$. Further for a set $I'\subseteq I$ of bins and a solution $S=(S_1,\dots,S_m)$ let $p^*(I') = \sum_{i\in I'} p^*(S_i)$. We define the efficiency $e_i$ of bin $i$ to be $e_i:=p_i/d_i$. 

We define the algorithm $\ALG^*$ in Figure~\ref{alg:algStar} for the modified \prob{Bin Covering} problem and denote as usual with $\ALG^*(I,J)$ the value of its solution for the modified problem on the instance $(I,J)$. Analogously denote $\OPT^*(I,J)$ the value of an optimal solution to the modified \prob{Bin Covering} problem.

\begin{figure}[h]
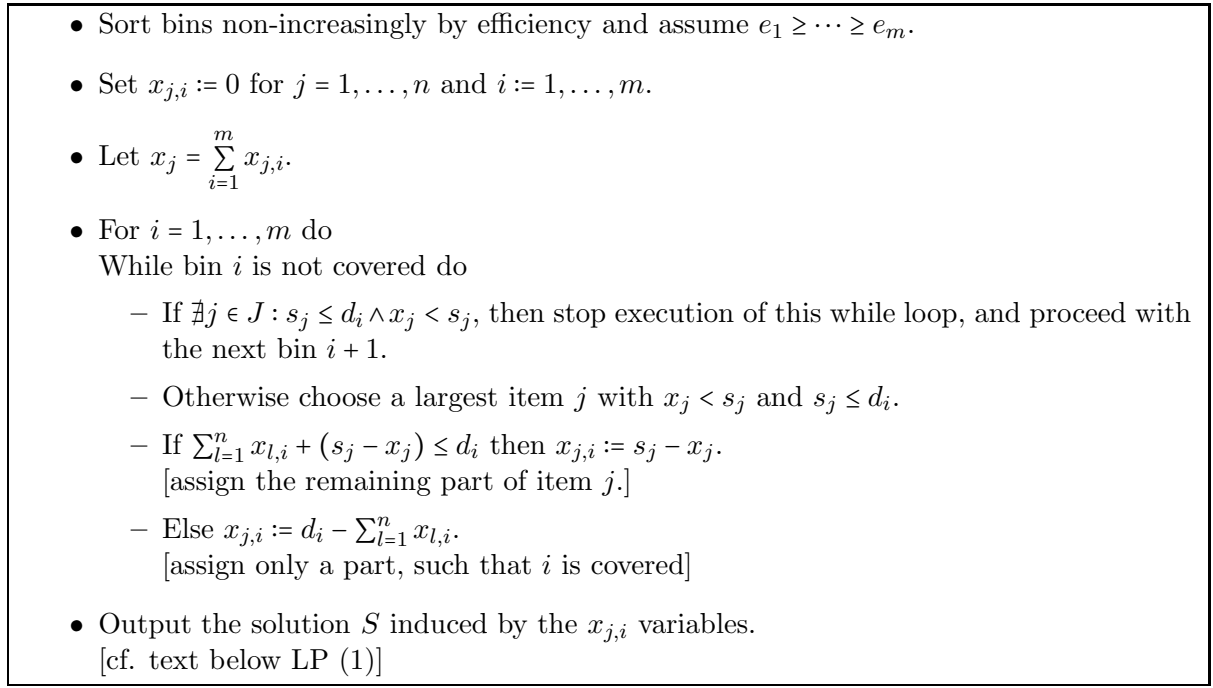

 \fbox{
 \begin{minipage}{\textwidth}
\begin{itemize}
\item Sort bins non-increasingly by efficiency and assume $e_1 \geq \dots \geq e_m$.
\item Set $x_{j,i} := 0$ for $j=1,\dots,n$ and $i:=1,\dots, m$.
\item Let $x_j = \sum \limits_{i=1}^m x_{j,i}$.
\item For $i=1,\dots,m$ do   \\
       While bin $i$ is not covered do 
    \begin{itemize} 
      \item If $\nexists j\in J: s_j \leq d_i  \wedge x_j < s_j$, then stop execution of this while loop, and proceed with the next bin $i+1$.
      \item Otherwise choose a largest item $j$ with $x_j < s_j$ and $s_j \leq d_i$.
       \item If $ \sum_{l=1}^n x_{l,i} +(s_j-x_j) \leq d_i$ then $x_{j,i} :=s_j-x_j$.\\ \ [assign the remaining part of item $j$.]
       \item Else $x_{j,i} := d_i- \sum_{l=1}^n x_{l,i}$. \\ \ [assign only a part, such that $i$ is covered]
      \end{itemize}
\item Output the solution $S$ induced by the $x_{j,i}$ variables. \\ \  [cf. text below LP~(\ref{lp:modifiedBinCovering})]
\end{itemize}
\end{minipage}
 }
\caption{The algorithm $\ALG^*$.}
\label{alg:algStar}
\end{figure}

Our modified \prob{Bin Covering} problem can be formulated as a linear program (LP), which will simplify the description of the analysis of the upcoming algorithm for the modified problem. Note that we do not need to solve this linear program. Actually, $\ALG^*$ solves it optimally. Moreover, Lemma~\ref{lem:relaxedProblemOptimal} and Lemma~\ref{lem:transformSolution1} imply a integrality gap of two for this linear program and the corresponding integer linear program.

\begin{alignat}{2}
\text{maximize} \quad &  \sum_{i=1}^m p_iy_i && \label{lp:modifiedBinCovering} \\
\text{subject to } \quad      & y_i \leq  \sum_{j=1}^n x_{j,i}/d_i \quad  && \forall i \in I  \notag \\ 
     & \sum_{i=1}^m x_{j,i} \leq s_j && \forall j\in J \notag  \\
  & 0 \leq y_i \leq 1 && \forall i\in I\notag    \\
 & 0 \leq x_{j,i}  && \forall i \in I  \forall j \in J\notag    \\
 & 0 \geq x_{j,i} && \forall i \in I  \forall j \in J\text{ with } s_j > d_i\notag    
\end{alignat}


We may identify the $x_{j,i}$ values with the sizes $s_{j,l}$, where $x_{j,i}=0$ means in fact that no part of item $j$ was assigned to bin $i$. If there are $p$ values $x_{j,l}>0$ for a fixed $j$ then this means that $\alg^*$ splits the item $j$ into $p$ parts, and there are $p$ values $s_{j,1}, \dots, s_{j,p} > 0$.

\begin{lemma}\label{lem:relaxedProblemOptimal} Algorithm $\ALG^*$ gives a solution of value $\ALG^*(I,J)=\OPT^*(I,J)$. 
\end{lemma}
\begin{proof}
We show that our algorithm gives an optimal solution to LP~(\ref{lp:modifiedBinCovering}) by transforming an arbitrary optimal solution to LP~(\ref{lp:modifiedBinCovering}) into a solution found by our algorithm without losing any profit. Let $O=(O_1,\dots,O_m)$  be an optimal solution and assume $(y_i',x_{j,i}')_{1\leq i \leq m, 1 \leq j \leq n}$ are the corresponding variables describing the assignment of the item parts to bins in $O$. Let $S$ be the solution found by $\alg^*$ and let  $(y_i,x_{i,j})_{1\leq i \leq m, 1 \leq j \leq n}$ be the corresponding variables set by $\alg^*$. In iteration $i=1,\dots,m$ we set $y_i':= y_i$ and $x_{j,i}' := x_{j,i}$ for all $j\in J$ and show that the optimality is preserved.   

Since items are arbitrary splittable and the assumption that all items are assigned by $\alg^*$ and $\opt$ we can assume that $s(S_i),s(O_i)\leq d_i$. Consider bin $i$ of the optimal solution and assume the bins $1,\dots,i-1$ in the optimal solution contain only the items, which are assigned to these bins by the solution of $\alg^*$, i.\,e. we have already $x_{j,l}'=x_{j,l}$ for all $1\leq l \leq i-1$ and all $j\in J$. 

Case $s(S_i) < d_i$:  This means that all items being admissible to bin $i$ were assigned to bin $i$ or to bins with indices $1,\dots, i-1$ in the solution $S$ by construction of the algorithm $\alg^*$. Formally we have thus $x_{j,i'}=0$ for all $j\in J$ with $s_j\leq d_i$ and all $i' > i$. Since for all variables $x_{j,i'}'=x_{j,i'}$ for all $i'< i$ and $j\in J$ holds, it follows that $\sum_{j=1}^n x_{j,i}' \leq \sum_{j=1}^n x_{j,i}$. If even $\sum_{j=1}^n x_{j,i}' < \sum_{j=1}^n x_{j,i}$ then  the ``missing'' item parts must be assigned to bins $i'>i$ in the optimal solution $O$. And so, if $x_{j,i}' < x_{j,i}$ for some $j$, we set $x_{j,i}':= x_{j,i}$ and $x_{j,i'}' := 0 $ for all $i' > i$. We increase the $y_i$ variable and decrease the $y_{i'}$ for $i'> i$ corresponding to the changes. By this process the objective value can only be increased, since $e_i \geq e_{i'}$ for all $i' > i$ and the sum of $y_i$ variables maintains the same. Also note that no constraints are violated, since $\alg^*$ assigns items only to bins, to which they are admissible.

Case $s(S_i) = d_i$: If $s(O_i) < d_i$ then again the ``missing'' items have to reside on bins $i' > i$ and we can increase some $x_{j,i}$ variables as above such that $\sum_{j=1}^n x_{j,i}' = \sum_{j=1}^n x_{j,i}$ holds. Now, if $x_{j,i}' \neq x_{j,i}$ for some $j\in J$ then there have to be items $j,j' \in J$ such that we have for the corresponding variables $x_{j,i}>x_{j,i}'$ and  $x_{j',i}<x_{j',i}'$. For the transformation we proceed in the ordering as $\alg^*$ did: Let $j_1,\dots,j_k$ the items assigned to bin $i$ by $\alg^*$ in this ordering, i.\,e. $\alg^*$ changed the values of the variables $x_{j_1,i},\dots, x_{j_k,i}$ in this ordering. Then we find the item $j_l \in \{j_1,\dots,j_k\}$ with smallest index $1\leq l \leq k$ such that $ x_{j_l,i}>x_{j_l,i}'$ holds and an arbitrary item $j'$ such that $x_{j',i}<x_{j',i}'$ holds. For ease of notation we set $j:= j_l$. 

Let $i'>i$ be the index of a bin such that $x_{j,i'}' > 0$. Clearly, such an index must exist by our assumption $x_j'= x_j=s_j$ holds and since  $x_{j,i'}' =x_{j,i'}$ for all $i' < i$ and all $j\in J$. As $x_{j,i} > 0$ and $x_{j,i'}' > 0$ it follows that item $j$ is admissible to $i$ and $i'$. 

Assume we already have shown $s_j \geq s_{j'}$. Note, this is not a priori clear, since $\alg^*$ may split items arbitrarily. Then, as  $x_{j,i'}' > 0$ item $j$ is admissible to bin $i'$ and so is the item $j'$, because of $s_j \geq s_{j'}$. Let $\delta:= \min \{ x_{j,i'}',x_{j',i}'\}$. We set $x_{j,i}':= x_{j,i}'+\delta$, $x_{j,i'}':= x_{j,i'}'-\delta$, $x_{j',i}':= x_{j',i}'-\delta$ and $x_{j',i'}':= x_{j',i'}'+\delta$. This process of adjusting can be iterated until $x_{j,i} = x_{j,i}'$ for all $j\in J$. 

We are left to show $s_j \geq s_{j'}$. We argue this comes indeed from the fact that the algorithm $\alg^*$ considers items in non-increasing order of size. Assume to the contrary that $s_{j'}> s_j$. Then $\alg^*$ would have considered item $j'$ before item $j$. Because of $x_{j',i}' > 0$ it must be in fact the case that $\alg^*$ assigns a part of item $j'$ to bin $i$ otherwise we had $x_{j',i}'=0$ by the property $x_{j,i'} =x_{j,i'}'$ for all $i' <i$ and all $j\in J$. 

Then either $x_{j',i} := s_{j'}-x_{j'}$ was set (in case $ \sum_{l=1}^n x_{l,i} +(s_j-x_j) \leq d_i$), which intuitively means that the yet unassigned rest of the item $j'$ was assigned to bin $i$ by $\alg^*$. In this case $x_{j',i'}= 0$ for all $i'> i$ and hence it follows $x_{j',i}' \leq x_{j',i}$, again because of the property $x_{j',i'}' =x_{j',i'}$ for all $i' < i$. This is a contradiction to the assumption $x_{j',i}' > x_{j',i}$.

Hence it must have been the case that $x_{j,i} := d_i- \sum_{l=1}^n x_{l,i}$ was set by the algorithm. We had $x_{j'} + \sum_{l=1}^n x_{l,i} > d_i$ in this situation, i.\,e. not the whole remaining size of item $j'$ was assigned to bin $i$. Then we conclude $j'$ was the last item assigned to bin $i$ by $\alg^*$. This contradicts the fact that item $j$ is assigned to bin $i$ by $\alg^*$, since we have assumed that $x_{j,i} >x_{j,i}' \geq 0$. This concludes the proof of the lemma.
\end{proof}

\begin{observation}\label{obs:feasibleSolution}
Let $S$ be a solution with the property $s(S_i) \leq d_i$ for all $i\in I$. Let $S^*$ be a solution, with the property $s(S^*_i) \leq 2d_i$ for all $i\in I$. If for all item parts $j\in J$ there holds $j\in S_i$ and $j \in S_{i'}^*$, where $e_i \leq e_{i'}$, then $p^*(S) \leq 2 p^*(S^*)$.
\end{observation}
\begin{proof}
If $s(S_i) \leq d_i$ for all $i\in I$ then we can compute $p^*(S)$ ``itemwise''

\begin{align}
p^*(S) & = \sum_{i \in I} s(S_i) e_i = \sum_{i \in I} \sum_{j \in S_i}  s_j e_i \notag  \\
       &\leq \sum_{i \in I} \sum_{j \in S^*_i}  s_j e_i = \sum_{i \in I}  s(S^*_i) e_i \label{eq:byPrec} \\
       & \leq  \sum_{i \in I} 2 \min \{ s(S^*_i)/d_i , 1\}d_ie_i \label{eq:boundLostProfit} \\
       & =  2 \sum_{i \in I}  p^*(S_i^*) = 2p^*(S^*), \notag
\end{align}
where Inequality~(\ref{eq:byPrec}) is by $e_i \leq e_{i'}$ for $i>i'$ and the fact that we assign items only to bins with smaller indices in $S^*$ in comparison to solution $S$ and Inequality~(\ref{eq:boundLostProfit}) holds, since by precondition we have for each $i\in I$ that $\sum_{j \in S^*_i}  s_j/2 \leq d_i$ in $S^*$. Hence the claim follows.
\end{proof}

We call a solution $S$ maximal with respect to the modified \prob{Bin Covering} problem, if there are no two distinct bins $i$ and $i'$ with $0< s(S_i) < d_i$ and $0<s(S_{i'}) <d_{i'}$ and $e_i \leq e_{i'}$, such that there is an item $j\in S_{i}$, which is admissible to bin $i'$. Note that this implies the following for such bins $i$ and $i'$. If we assign in a maximal solution only one item from a bin $i$ to a bin $i'$ then bin $i'$ is already covered by only this item. We say a solution $S$ contains no split items, if for all $i\in I$ and $j\in S_i$ we have $s_{j,1}=s_j$.

\begin{lemma}\label{lem:transformSolution1}
Let $S$ be a solution given by $\ALG^*$ for the modified problem. Then there exists a solution $S^*$ such that $p^*(S) \leq 2p^*(S^*)$ and $S^*$ contains no split items. Further $S^*$ is maximal with respect to the  modified \prob{Bin Covering} problem.
\end{lemma}
\begin{proof}
In a first step we create a solution $S'$ such that $S'$ contains no split items, i.\,e. $s_{j,1} = s_j$ for all $j\in J$. Let $j$ be an item, which is split by $\ALG^*$ into $p>1$ parts $(j,1), \dots (j,p)$. We ``merge'' the parts $(j,1),\dots,(j,p)$ the following way. We assign all the parts $(j,2),\dots, (j,p)$ to the bin, to which $(j,1)$ was assigned. The solution created is the solution $S'$. 

We argue that by this procedure each bin receives parts of at most one item. By the algorithm, an item $j$ is split only then, when it fills a bin. The bin filled is assigned a part $(j,l)$ and never receives an item (part) after that. Hence each bin $i$ receives at most one first part of an item, i.\,e. if $(j^*,1) \in S_i$ then $(j,1)\notin S_i$ for $j\neq j^*$. It follows that each bin is assigned the parts of only one item, since parts $(j,l)$ are always assigned to the bin $i$ with $(j,1)\in S_i$. 

Assume item $j$ was split into $p>1$ parts $(j,1),\dots,(j,p)$ by $\ALG^*$. If a part $(j,1)$ was assigned to a bin $i$ then part $(j,1)$ was admissible and so was item $j$. Thus $s_j \leq d_i$. It follows that $s_{j,1}+\dots +s_{j,p} \leq d_i$. As with the above argumentation each bin $i$ receives only parts of one item $j$ and we had $s(S_i) \leq d_i$, by the algorithm $\ALG^*$, it follows $s(S'_i) \leq s(S_i)-s_{j,1}+s_{j,1}+s_{j,2}+\dots +s_{j,p}< 2d_i$, because $s_{j,1} >0$. Hence we have shown that $s(S'_i)< 2d_i$ for each bin $i\in I$ in the solution $S'$. 

As already described, when $\ALG^*$ splits an item $j$ into $p>1$ parts, then the first part $(j,1)$ fills a bin. Since $\ALG^*$ considers the items in non-increasing order of efficiency, we have that the parts $(j,2),\dots,(j,p)$ are assigned to bins with at most the same efficiency. Hence in the solution $S'$ each item is assigned to a bin $i$ with at least the same efficiency as the bin $i'$, to which it was assigned in the solution $S$.

The solution $S'$ can now be transformed into the maximal solution $S^*$ preserving the both mentioned properties, which are the preconditions of Observation~\ref{obs:feasibleSolution}. For this let $T :=\{ i \in I \mid 0 < s(S'_i) < d_i \}$. Let $T = \{e_{i_1},\dots,e_{i_l}\}$ and assume as usual $e_{i_1} \geq \dots \geq e_{i_l}$. For $j=1,\dots,l$ do the following. While bin $i_j$ is not covered and one of the bins $i_{j+1},\dots, i_l$ contains an item $j'$, which is admissible to $i$, assign $j'$ to bin $i_j$. If $i_j$ is covered or there are no items left on the bins $i_{j+1},\dots, i_l$, which are admissible to $i_j$, then proceed with bin $i_{j+1}$ and so on. The so created solution is the solution $S^*$.

Clearly, items are only assigned to more efficient bins. Further, an item $j$ is only assigned to a bin $i$, when $j$ is admissible to $i$ and $i$ is not yet covered. Hence, it follows $s(S_i^*) \leq 2d_i$. Moreover, if $s(S_i^*)< d_i$ for a bin $i$, then by construction all bins $i'\in T$ with $e_i'\leq e_i$ do not contain an item $j'$ which is admissible anymore. Hence $S^*$ is also maximal w.\,r.\,t. the modified \prob{Bin Covering} problem. Applying Observation~\ref{obs:feasibleSolution} to the solutions $S$ and $S^*$ gives the claim of the lemma.
\end{proof}

\begin{lemma}\label{lem:transformSolution2} Let $S$ be a solution being maximal with respect to the modified \prob{Bin Covering} problem and containing no split items. Then there exist a solution $S^*$ for the \prob{Generalized Bin Covering} problem, such that $p^*(S) \leq 2p(S^*)$.
\end{lemma}
\begin{proof}
Let $R := \{ i \in I \mid 0< s(S_i) < d_i\}$. Assume \wlgs that $R=\{1,\dots, l\}$ and $e_1\geq \dots \geq e_l$. Construct two solutions $S'$ and $S''$. We set $S_l'=J$ and $S_i' = \emptyset$ for $1 \leq i \leq m, i \neq l$. Further we set $S''_{i-1} = S_{i}$ for $2\leq i \leq l$, $S''_l = \emptyset$ and $ S''_i := S_i$ for $l < i \leq m$.  

In $S'$ the only bin $l$ which is assigned items is covered, since we may assume \wlgs that each bin $i$ is covered, when all items are assigned to it. The same holds true for $S''$: The bins from $\{l+1,\dots,m\}$, which contain items, are covered since they were already covered in $S$. The bins $1,\dots,l-1$ are covered, since $S$ is a maximal solution with respect to the modified \prob{Bin Covering} problem and $e_i \geq e_{i+1}$ holds for $1 \leq i \leq l-1$. Finally we have $S''_l = \emptyset$.

As none of the solutions $S'$ and $S''$ contains split items and all bins are covered, we have that $p^*(S') =p(S')$ and $p^*(S'') =p(S'')$. We output $S^* := S'$ if $p(S')=\max\{p(S'), p(S'')\}$ and $S^*:= S''$ otherwise.

To see the claim about the approximation guarantee distinguish the cases $p_l > p^*(S)/2$ and $p_l \leq p^*(S)/2$, where the index $l$ is as above the index of the bin with $S'_l= J$. If $p_l > p^*(S)/2$ then $p(S^*)\geq p(S') = p_l > p^*(S)/2$. If $p_l \leq p^*(S)/2$ then $p(S^*) \geq  p(S'') = p^*(S) - p_l \geq  p^*(S) - p^*(S)/2 = p^*(S)/2$, which concludes the proof of the lemma.
\end{proof}

\begin{proof2}[of Theorem~\ref{thm:gbc5approx}]
Let $(I,J)$ be the given instance. Our algorithm works as follows. We use Observation~\ref{obs:matchingSingular} to find a solution $S_1$. Then we run $\ALG^*$ on the instance $(I,J)$ and let $S$ be the solution output. We transform the $S$ into a solution $S'$ as done in Lemma~\ref{lem:transformSolution1} and then solution $S'$ into solution $S_2$ as done in Lemma~\ref{lem:transformSolution2}. We output the better solution from $\{S_1,S_2\}$. The running time is clearly dominated by the algorithm for \prob{Maximum Weight Bipartite Matching}. 

We give the proof on the approximation guarantee. Fix an optimal solution $O$ to the instance $(I,J)$. Let $I_R \subseteq I$ be the set of bins covered regularly by the solution $O$ and $J_R=\{j\in J \mid \exists i\in I_R: j\in O_i \}$, the set of items on these bins. Let $I_S \subseteq I$ be the set of bins covered singularly by the solution $O$ and $J_S=\{j\in J \mid \exists i\in I_S: j\in O_i \}$, the set of items on these bins. We have $\OPT(I,J)  = \OPT(I_R,J_R) + \OPT (I_S,J_S)$. Thus $\OPT(I,J) - \OPT(I_R,J_R) = \OPT (I_S,J_S)$.

Case $\OPT(I_R,J_R) < 4/5 \cdot \OPT(I,J)$. Then $\OPT(I_S, J_S) > 1/5 \cdot \OPT(I,J)$ by the above. Hence in this case we output a solution such that $\OPT(I,J) \leq 5\ALG(I,J)$ by Observation~\ref{obs:matchingSingular}.

Case $\OPT(I_R,J_R) \geq 4/5 \cdot \OPT(I,J)$. It follows $\OPT (I_S,J_S) \leq 1/5 \cdot \OPT(I,J)$. We find
\begin{align}
\OPT(I,J) & = \OPT(I_R,J_R) + \OPT (I_S,J_S) \notag \\
	  & \leq \OPT^*(I_R,J_R) +  1/5\cdot  \OPT (I,J) \label{eq:gbc_step1} \\
	  & \leq \OPT^*(I,J)+ 1/5\cdot  \OPT (I,J) \notag  \\
	  & = \ALG^*(I,J)+ 1/5\cdot \OPT (I,J) \label{eq:gbc_step2} \\
	  & \leq 4\ALG(I,J)+1/5\cdot \OPT (I,J). \label{eq:gbc_step3}
	  \end{align}

In Inequality~(\ref{eq:gbc_step1}) we use $\OPT^*(I_R,J_R) \geq \OPT(I_R,J_R)$ and the assumption of the case. In Inequality~(\ref{eq:gbc_step2}) we use Lemma~\ref{lem:relaxedProblemOptimal}. In Inequality~(\ref{eq:gbc_step3}) we have accounted for transforming the fractional solution to the modified problem into a solution for the \prob{Generalized Bin Covering} problem with Lemmas~\ref{lem:transformSolution1} and~\ref{lem:transformSolution2}. It follows $\OPT(I,J) \leq 5\ALG(I,J)$.
\end{proof2}

\section{Variable-Sized Bin-Covering}
\label{sec:variable-sized_unit_supply}

\subsection{A Tight Analysis of Next Fit Decreasing in the Unit Supply Model}
\label{sec:nfd}

In this subsection, we have $d_i = p_i$ for all $i$ in the unit supply model. The algorithm \prob{Next Fit Decreasing} (\NFD) is given in Figure~\ref{alg:nfd}. The algorithm considers bins in non-increasing order of demand. For each bin, if the total size of the unassigned items suffices for coverage, it assigns as many items (also non-increasing in size) as necessary to cover the bin. Otherwise, the bin is skipped. In this section we assume that we have $d_1\geq \dots \geq d_m$ and $s_1 \geq \dots \geq s_n$, as needed by the algorithm. 

\begin{figure}[h]
\fbox{
\begin{minipage}{\textwidth}
\begin{itemize}
  \item Sort bins non-increasingly by demand and rename the bins such that $d_1 \geq \dots \geq d_m$.
  \item Sort items non-increasingly by size and rename the items such that $s_1 \geq \dots \geq s_n$.
  \item Let $i=1$ be the current bin and $j=1$ the index of the first unassigned item.
  \item While $j \leq n$ and $i\leq m$ do 
    \begin{itemize}
    \item If $\sum_{l=j}^n s_l < d_i$ set $i:= i+1$ and $S_i = \emptyset$.
    \item Else let $j'$ be the smallest index with $\sum_{l=j}^{j'} s_l \geq d_i$. \\
	Assign the items $j,\dots, j'$ to bin $i$, i.\,e., $S_i = \{j, \dots, j'\}$ \\
	Set $i:= i+1$ and $j=j'+1$.
    \end{itemize}
\item Return $S = (S_i)_{i \in I}$.
\end{itemize}
\end{minipage}}

\caption{Algorithm \NFD.}
\label{alg:nfd}
\end{figure}

\begin{example}
Let $2/3> \epsilon > 0$ be arbitrary. The following instance $(I,J)$ yields that \NFDs gives an approximation not better than $9/4-2\epsilon$. Hence \NFDs is at least a $9/4$-approximation. Let $I = \{4,3-2\epsilon,3-2\epsilon,3-2\epsilon\}$ and $J=\{2-\epsilon,2-\epsilon,2-\epsilon,1-\epsilon,1-\epsilon,1-\epsilon\}$. Observe we have $\NFD(I,J) = 4$ and $\OPT(I,J) =9-6\epsilon$.
\label{exa:nfd_lower_bound}
\end{example}

\begin{theorem}\label{thm:NFDisSuperb}
\NFDs is a 9/4-approximation algorithm with running time $O(n \log n + m \log m)$. The bound is tight.
\end{theorem}

Note that this is almost best possible, since the problem is inapproximable up to a factor of two, unless $\P = \NP$, which holds for unit supply even asymptotically in the notion of Theorem~\ref{thm:lower_bound}.

\noindent
{\bf Proof techniques.} We will use three kinds of arguments. The first type we call a volume argument. If $s$ is the sum of item sizes in the (remaining) instance, we have $\OPT \leq s$. This argument holds independently of the concrete demands of bins. Such volume bounds are too weak in general to achieve the claimed bound, thus we need arguments using the structure of bins in the instance, which is the second type of arguments. For example, if the sum of item sizes in the (remaining) instance is $\alpha d$, $\alpha>1$ and the demand of the only bin in the instance is $d$, then it follows $\OPT \leq d$, while we could only conclude $\OPT\leq \alpha d$ with a volume argument. The third type of argument we use are arguments transforming instances. These arguments give that we can \wlgs restrict ourselves to analyze instances having certain properties. For example, we may assume that there are no items in the instance with size larger than the largest bin demand.\\
\noindent
{\bf Proof outline.} Our proof looks at the specific structure of the solution given by \NFDs and argues based on that, how much better an optimal solution can be. We employ the described techniques in the following way. Firstly, we settle two basic properties of \NFD: A solution of \NFDs is unique (Observation~\ref{obs:i0}) and if a bin is covered with at least twice its demand, then there is only one item assigned to it (Observation~\ref{obs:moreThanTwiceDemand}). These properties will be used implicitly during the analysis. After that, we give transformation arguments, which allow us to restrict ourselves to analyze instances with the following properties. We may assume that \NFDs covers the first bin (Observation~\ref{obs:useFirstBin}), and that the ``right-most bins'' (i.\,e. the bins with the least demand -- or the smallest bins) are empty (Observation~\ref{obs:reduce2}), where we will specify this notion in more detail later. We will show that we may assume that the ``left-most bins'' (i.\,e. the largest bins) are only assigned items such that they do not exceed twice their demand (Lemma~\ref{lem:everyBigJobIsAOptBigJob}). Here ``left-most bins'' refers to the bins up to the first empty bin.

With these tools at hand we can come to the actual proof. The central notion here is the {\em well-covered} bin (Definition~\ref{defi:well-covered}): Consider the right-most (i.e., smallest) empty bin in the instance with the property that all larger bins are assigned items only up to twice their demand. If such a bin exists, then we call the covered bins of these well-covered. The proof will be inductive. The terminating cases are the ones, when there are either at least four well-covered bins (Observation~\ref{obs:bigSolutionsGood}) or between two and three well-covered bins, but there is a bin among these containing at least three items (Lemma \ref{lem:twoOrThreeMachinesBeforeGap}). These cases are settled by volume arguments which is the reason, why they are terminating cases -- even if there are additional filled but not well-covered bins in the instance. We are also in terminating cases if the above prerequisites are not met, but there are no filled bins which are not well-covered: Lemma~\ref{lem:wellCoveredWithAtMostTwoJobs} treats the case that all of the at most three well-covered bins contain at most two items and Lemma~\ref{lem:onlyOneMachineBeforeGap1} gives the cases, in which we have exactly one well-covered bin in the instance.

If there are additional filled but not well-covered bins and we cannot apply volume arguments -- as in the both last mentioned situations -- , we have to look at the instance more closely. Our idea is here to consider a specific not well-covered bin, which will be called the {\em head of the instance}. We will subdivide such an instance into two parts, which is done by the key lemma of the recursion step, the Decomposition Lemma~\ref{lem:decompositionLemma}. Therein and in Lemma~\ref{lem:onlyOneMachineBeforeGap2} we show that it is not advantageous to assign items, which \NFDs assigned to bins with larger demand than the demand of the head of the instance, to bins with smaller demand than the demand of the head of the instance. This allows us in combination with some estimations to consider the left part of the instance and the right part separately. For the left part Lemma~\ref{lem:atMostTwoJobsOnIPrime} and Lemma~\ref{lem:decompositionLemma} give that the approximation factor of \NFDs is at most $9/4$ and the right part of the instance is a smaller instance and we may hence iteratively apply the argumentation.

We now start with the proof and give four observations, which are easy but also rather important and will be used often implicitly during the analysis. The first observation is immediate from \NFD's behavior, and thus no proof is necessary. 

\begin{observation}\label{obs:i0} Fix an instance $(I,J)$. Then the solution of \NFDs (up to renaming items of identical size) for this instance is unique. Further, if \NFDs did not assign the item $j_0$ to a bin, then \NFDs does not assign the items $j_0+1,\dots,n$ to a bin either. 
\end{observation}

In the following we similarly assume \wlgs that if for any bin $i$ we have $u_\OPT(i) > 0$ then $u_\OPT(i) \geq d_i$ and if $S_i$ is the set of items assigned to a bin $i$, then $u_\OPT(i) -s_j < d_i$ for any $j \in S_i$, i.\,e. \OPT\ does not assign items, which are not needed to cover a bin. The next observation gives that, if a bin in \NFD's solution is assigned at least twice its demand, there is only one item on it. Recall $S_i$ is the set of items \NFD\ assigns to bin $i$.

\begin{observation}
\label{obs:moreThanTwiceDemand}
If $s(S_i) \geq 2 d_i$ then $|S_i| = 1$.
\end{observation}
\begin{proof}
Assume $|S_i| \geq 2$. Immediate from the algorithm we have that \NFDs uses a new bin, if the current bin is full and \NFDs never puts an item on a full bin. Hence let there be only one item $j^*$ above the boundary of $d_i$, i.\,e. if $1,\dots,j^*$ are the items, which \NFDs assigns to bin $i$, then we have $u(i)-s_{j^*} < d_i$ and $u(i) \geq d_i$. If $u(i) \geq 2 d_i$, then it follows that $s_{j^*} > d_i$. But for all items $j \in S_i \setminus \{j^*\}$ it holds $s_j< d_i$, because this even holds for the sum of all these items. As $|S_i|> 1$ this is a contradiction to the fact that \NFDs assigns items in non-increasing order.
\end{proof}

The next observation gives that we may restrict ourselves to the analysis of such instances, where \NFDs covers the first bin.

\begin{observation}\label{obs:useFirstBin}
Fix an instance $(I,J)$. If $u_{\NFD}(1)=\dots=u_{\NFD}(i) = 0$ then $u_{\OPT}(1)=\dots=u_{\OPT}(i) = 0$. Let $I'' = I\setminus \{1,\dots,i\}$. Then $\NFD(I,J) = \NFD(I'',J)$ and $\OPT(I,J) = \OPT(I'',J)$.
\end{observation}
\begin{proof}
If $u_{\NFD}(1)=\dots=u_{\NFD}(i) = 0$, then from \NFD's behavior we know that the sum of all item sizes in the instance did not suffice to fill bin $i$. Hence, by the ordering of bins, it does not suffice to fill the bins $1,\dots,i-1$ either. Of course, the same holds true for \OPT. \end{proof}

By the argument given by the previous observation it is also justified to assume $\NFD(I,J) > 0$. Since otherwise also $\OPT(I,J) = 0$ follows and \NFDs is optimal. This assumption will always be implicitly used and thus the quotient $\OPT(I,J)/ \NFD(I,J)$ is always defined. Alternatively we could define $\OPT(I,J)/ \NFD(I,J):=1$, if $\NFD(I,J)=\OPT(I,J)=0$.

Further, we may always assume that there exists an empty bin, otherwise \NFDs is clearly optimal. We may strengthen this observation such that it suffices to compare instances of \NFDs to \OPT, where the right-most bins are all empty, i.\,e. there is a non-empty bin $i'$, the bin $i'+1$ is empty and all bins with higher indices, if they exist, are also empty. 

\begin{observation}\label{obs:reduce2}
Let a solution of \NFDs for an instance $(I,J)$ be given. Let $i^*$ be a bin with $u(i^*) = 0$ and for all $i > i^*$ we have $u(i) > 0$. Then $\OPT(I,J)/ \NFD(I,J) \le \OPT (I',J) / \NFD(I',J)$, where $I' = I \setminus \{ i^* + 1, \dots, m\}$.
\end{observation}
\begin{proof}
Since \NFDs did not fill bin $i^*$ we know by \NFD's behavior $d_{i^*} > \sum_{l= i^*+1}^m d_l$. Hence on the one hand for the instance $(I',J)$ we have $\NFD(I',J) = \NFD(I,J) - \sum_{l= i^*+1}^m d_l$ by \NFD's behavior. As we may assume that bin $1$ is filled, we have on the other hand $\NFD(I,J) \geq d_1 \geq d_{i^*} >\sum_{l= i^*+1}^m d_l$ and it follows $\NFD(I,J) >  \sum_{l= i^*+1}^m d_l$. For \OPT\ we have

$$ \OPT(I',J) \geq \OPT(I,J)- \sum_{\substack{l > i^*, \\ u(l) > 0}}  d_l,$$

since \OPT\ can possibly assign the items, which potentially reside on the bins $i^*+1,\dots,m$ in its solution to the instance $(I,J)$ in the instance $(I',J)$ to other bins.  Altogether we find

$$ \frac{\OPT (I',J) }{ \NFD(I',J)} \geq \frac{\OPT(I,J)- \sum_{\substack{l > i^*, \\ u(l) > 0}}  d_l}{\NFD(I,J) - \sum_{l= i^*+1}^m d_l} \geq \frac{\OPT(I,J)- \sum_{l= i^*+1}^m  d_l}{\NFD(I,J) - \sum_{l= i^*+1}^m d_l}  \geq \frac{\OPT(I,J)}{\NFD(I,J)}, $$

where the last inequality holds, because of $\OPT(I,J) \geq \NFD(I,J) > \sum_{l= i^*+1}^m d_l >0$.
\end{proof}

With this observation let $u(m) = 0$ from now on. The next lemma states that we can assume \wlgs that all bins $i$ up to the bin with smallest index $i^*$, such that $u(i^*+1)=0$, receive only items in such a way, that $u(i) \leq 2d_i$ for $i< i^*$.

\begin{lemma}\label{lem:everyBigJobIsAOptBigJob}
Let $(I,J)$ be an instance and consider a solution of \NFDs for it. Let $i^*$ be the smallest index, such that $i^*$ is a bin with $u(i^*) > 0 $ and $u(i^*+1) = 0$. Let $i_1,\dots,i_k \in\{1,\dots,i^*\}$ be the indices with $u(i_j) \geq 2d_{i_j}$ for $j=1,\dots,k$ and let $j_1,\dots,j_k$ be the items on these bins. Set $I'=I\setminus \{i_1,\dots,i_k \}$ and $J'=J\setminus \{j_1,\dots,j_k\}$. Then $\OPT(I,J)/\NFD(I,J) \leq  \OPT(I',J')/\NFD(I',J').$ 
\end{lemma}
\begin{proof}
Consider bin $i_1$ and item $j_1$, where we have by Observation~\ref{obs:moreThanTwiceDemand} that on $i_1$ in fact resides only one item. We argue that we can assume, an optimal algorithm also assigns $j_1$ to $i_1$.

By the ordering of items, for every bin $i\geq i_1$, if we assign $j_1$ to $i$, we have that $s_{j_1} -d_i \geq s_{j_1} - d_{i_1} = u(i_1)-d_{i_1}$. Because of this and because every bin $i \geq i_1$ is covered with only the item $j_1$ we can assume an optimal algorithm would assign $j_1$ to a bin with an index at most $i_1$. 

If an optimal algorithm decides to assign $j_1$ to a bin with index smaller than $i_1$, then it would not assign all of the items $1,\dots,j_1-1$ to the bins $1,\dots,i_1-1$. This is because also \NFDs assigns the items $1,\dots,j_1-1$ to the bins $1,\dots,i_1-1$ and these are already covered, but the bin $i_1$ would not, in this case. Hence at least one of the items $1,\dots,j_1$ would be assigned to a bin with index at least $i_1$, otherwise this assignment would not be optimal. With the same argumentation as above, for every such item $j$, if it would be assigned to a bin $i\geq i_1$ we had $s_j-d_i \geq s_j-d_{i_1}$. Because every such item $j$ covers every bin $i\geq i_1$ alone, we can thus assume $j$ is assigned to $i_1$. But then, since $s_j-d_{i_1} \geq  s_{j_1}-d_{i_1}$, we also can assume $j_1$ is assigned by \OPT\ to $i_1$. 

Hence, defining $I'' = I \setminus \{i_1\}$ and $J''= \setminus \{j_1\}$ we have 

$$ \frac {\OPT(I,J)}{  \NFD(I,J)} =  \frac{ d_{i_1} + \OPT(I'',J'') }{  d_{i_1} +\NFD(I'',J'') } \leq \frac {\OPT(I'',J'')}{ \NFD(I'',J'')}, $$

where the last inequality is because $\OPT(I'',J'') \geq  \NFD(I'',J'') > 0$. Iteratively applying this argumentation yields the statement.
\end{proof}

In order to simplify the analysis we define the notion of a {\em well-covered bin}.

\begin{definition}\label{defi:well-covered}
Consider a solution of \NFDs for an instance $(I,J)$. Fix a bin $i^*$, with $u(i^*) > 0$, and let $i'$ be the smallest number with $i' > i^* $ such that $u(i') = 0$, if it exists. We call the bin $i^*$ well-covered, if $i'$ exists and $u(i) \leq 2d_i$ for all $i=1,\dots,i'$.
\end{definition}

By Observation~\ref{obs:useFirstBin} and Lemma~\ref{lem:everyBigJobIsAOptBigJob} it can be shown that we may assume that there is at least one well-covered bin in the instance.

\begin{observation}\label{obs:wellCovWellDefined}
Consider a solution of \NFDs for an instance $(I,J)$, which contains at least one filled bin. The number $k$ of well-covered bins is well-defined and we can assume $k\geq 1$.
\end{observation}
\begin{proof}
We may assume there is an empty bin and let $i'$ be the smallest index with $u(i') = 0$. By Observation~\ref{obs:useFirstBin} we may assume the bin $1$ is filled and thus $i' > 1$.

By Lemma~\ref{lem:everyBigJobIsAOptBigJob} we may assume that for the bins $i=1,\dots,i'$ we have $u(i) \leq 2d_i$. Let $k'$ the largest index, such that $u(k'+1) = 0$ and $u(k') > 0$ and $u(i) \leq 2d_i$ for all $i=1,\dots,k'$. Let $k = |\{ i \in \{ 1,\dots,k' \} \mid u(i) > 0\}|$. Hence $k'>1$ exists and is well-defined as the solution of \NFDs is unique for a given instance. It follows $k$ is well-defined and $k\geq 1$ as $u(1) > 0$.
\end{proof}

We already can conclude that \NFDs is at most a $3$-approximation. 

\begin{observation} \label{obs:bigSolutionsGood}
Let $(I,J)$ be given. If \NFDs gives a solution with $k$ well-covered bins, then $\OPT(I, J)/\NFD(I, J) \leq 2+1/k$.
\end{observation}
\begin{proof}
Let $k'$ be the largest index of a well-covered bin and let $I'= \{ i \in \{1,\dots,k'\} \mid u(i) > 0\}$, be the set of the well-covered bins. On the one hand we have $\NFD(I,J) \geq \sum_{i\in I'} d_i$ and on the other $\NFD(I,J) \geq kd_{k'}$.

Recall, we have for every $i\in I' $ that $u(i) \leq 2d_i$. Let $l$ be the index of the first item, which \NFD\ could not assign to a bin with index $k'$ or smaller. Since $u(k'+1)=0$ by definition of $k'$ we further have $\sum_{j=l}^n s_j < d_{k'}$, otherwise \NFD\ would have filled bin $k'+1$. It follows for the sum of item sizes $s=\sum_{j=1}^n s_j < \sum_{i \in I'} 2d_i + d_{k'}$.

Because $\OPT \leq s$ we can bound $\OPT <  \sum_{i \in I'} 2d_i + d_{k'} \leq 2\NFD + 1/k \cdot \NFD = (2+1/k)\NFD$.
\end{proof}

For a number of $k\geq 4$ well-covered bins we already have the desired result. We now turn our attention to the cases, when $k\leq 3$.

\begin{lemma} \label{lem:wellCoveredWithAtMostTwoJobs}
Let $(I,J)$ be an instance such that \NFDs gives a solution, in which every filled bin is well-covered and contains at most two items. Then $\OPT(I,J)/\NFD(I,J) \leq 2$.
\end{lemma}
\begin{proof}
By Observation~\ref{obs:useFirstBin} we assume $u_{\NFD}(1) >0$. Call a maximal set of neighbouring, empty bins a gap in \NFD's solution, that is formally a set of bins $\{i,i+1,\dots,i+i'\}$ with $u(i)=\dots=u(i+i')=0$ and $u(i-1) > 0$ and $i+i'+1=m+1$ or $u(i+i'+1) >0$. Enumerate the gaps from left to right. For every gap $G_l$ consider now the set of filled bins $F_l$ before this gap, that is, if $i'$ is the bin with smallest index in $G_l$ and $i$ is the bin with highest index in $G_{l-1}$, then $F_l = \{i+1,\dots,i'-1\}$, where we assume $0$ is the bin with highest index in the ``gap'' $G_0$. 

If there are $t$ many items on the bins in $F_l$, then we modify the sets $G_l$ in such a way that $|F_l|+|G_l|=t$ for every $l$, if not already the case. If $|G_l| < t-|F_l|$, then introduce $t-|F_l|-|G_l|$ many bins with demand $d_i$ into $G_l$, where $i$ is the bin with highest index in $G_l$. If $|F_l|+|G_l| > t$ then remove the $|F_l|+|G_l|-t$ smallest bins from $G_l$ and observe that this is always possible, i.\,e. we have $|G_l| \geq |F_l|+|G_l|-t$, since we have $|F_l| \leq t$. We remark, it does not matter, if $|G_l|=0$ now, for some gap $l$.

Observe, by \NFD's behavior, if $I'$ is the set of bins in the modified instance, we not only have $\NFD(I',J) = \NFD(I,J)$ but also all items are assigned to the same bins in both solutions, which is because we copied and removed only unfilled bins. 

Note, if $i$ is an empty bin in \NFD's solution, and $S_{\leq i}$ is the set of items, which reside on a bin with index at most $i$ in \NFD's solution, then \NFDs can not fill a bin with index at most $i$ with only the items from $J \setminus S_{\leq i}$. Then, also \OPT\ can only fill a bin with index $i' \leq i$, if there is also an item from the set $S_{\leq i}$ on this bin. Since, if this would not be the case, then also \NFDs would have filled $i$, and this argumentation holds for both instances $(I,J)$ and $(I',J)$. 

Consider the following relaxed \prob{Bin Covering} problem. In order to yield the profit $d_i$ for a bin $i$ either the bin $i$ has to be covered (with some items) or an item $j$ from the set $S_{\leq i}$ has to be assigned to $i$. The problem is clearly a relaxation of the ordinary \prob{Bin Covering} problem. Let $\OPT^*$ denote an optimal algorithm for the relaxed problem. The value of this algorithm is obviously only an overestimation for \OPT\ on the same instance, that is $\OPT^*(I,J) \geq \OPT(I,J)$.

With this modified notion of \OPT\ we will show $\OPT^*(I',J) \geq \OPT^*(I,J)$ by construction and then, because of $\OPT^*(I,J) \geq \OPT(I,J)$, it suffices to bound $\OPT^*(I',J) \leq 2\NFD(I',J) = 2\NFD(I,J)$, where we have already explained the last equality above.

For $\OPT^*(I',J) \geq \OPT^*(I,J)$ we only have to justify that the removing of bins from the sets $G_l$ does not make $\OPT^*$ lose any profit on the instance $(I',J)$ in comparison to the instance $(I,J)$. Fix a set $F_l \cup G_l$ and consider a bin $i$ with an index smaller than that of any bin from $F_l \cup G_l$. Let $l' < l$ be such that $i\in F_{l'} \cup G_{l'}$ (before modifying the instance). On the one hand, recall that no subset of the items in $F_{l'+1} \cup F_{l'+2} \cup \dots$ can cover bin $i$ without items in $F_1 \cup \dots \cup F_{l'}$. On the other hand, if an item $j'$ from a bin from $F_1 \cup \dots \cup F_{l'}$ resides on bin $i$, then, by our relaxation, $\OPT^*$ already gains the profit for bin $i$, without assigning another item to this bin. In conclusion $\OPT^*$ would not assign such an item $j \in F_{l'+1} \cup F_{l'+2} \cup \dots$ to a bin from $F_{l'} \cup G_{l'}$.

Note also that assigning such an item $j$ to bins with larger index than that from $F_l \cup G_l$ can only be worse, by the ordering of bins by demand. To sum up, the profit of all bins in the modified instance $(I',J)$ will be gained by $\OPT^*$ with only one item and this is clearly optimal.

Hence considering a solution of $\OPT^*$ to the instance $(I,J)$ and a bin $i\in I$, which is no longer available in the instance $(I',J)$, i.\,e. $i\notin I'$, we can move all items $S_i$, i.\,e. the items from bin $i$ in the solution to $(I,J)$, and assign them to a bin from the set $F_l \cup G_l$ in the modified instance $(I',J)$. Thus removing bin $i$ does not have any effect on $\OPT^*$'s solution, since this bin would be empty anyway. As this argument holds for every bin removed, the claim $\OPT^*(I',J) \geq \OPT^*(I,J)$ follows.

Now, in order to establish $\OPT^*(I',J) \leq 2\NFD(I',J)$ we claim we can associate every bin from every gap $G_l$ to a bin with at least equal demand from $F_l$. This can be seen as follows. If $t$ is the number of items on the bins from $F_l$ in \NFD's solution, then we have $|F_l|+|G_l| = t$ in the modified instance and with $|F_l| \geq t/2$ as at most two items reside by prerequisite on every bin from $F_l$, it follows $|F_l| \geq |G_l|$. By construction every bin in $G_l$ has at most as much demand as the smallest bin from $F_l$ and the claim is proved.

As we have argued that $\OPT^*$ gains the profit for the bins in $F_l \cup G_l$ in the modified instance, while $\NFD$ gains the profit for the bins in $F_l$, we have $\OPT^*(I',J) \leq 2\NFD(I',J)$ and the statement follows. 
\end{proof}

\begin{lemma}\label{lem:onlyOneMachineBeforeGap1}
Let $(I,J)$ be an instance. If \NFDs gives a solution with $k=1$ well-covered bins and all other bins are empty, then $\OPT(I,J)/\NFD(I,J) \leq 9/4$.
\end{lemma}
\begin{proof}
We may assume that the well-covered bin contains at least three items, otherwise the claim follows immediately from Lemma~\ref{lem:wellCoveredWithAtMostTwoJobs}. Let $t$ be the number of items on bin $1$. Observe that if there are $t$ items on bin $1$ in \NFD's solution, we have that these have in sum a size of at most $t/(t-1) d_1$, which is by the ordering of items and the fact that \NFDs does not assign items to already covered bins. With the argument that the items $t+1,\dots,n$ could not fill any bins in $2,\dots,m$ in \NFD's solution, it follows we have $\sum_{l=t+1}^n s_l < d_2 \leq d_1$. 

Altogether we can bound $\OPT < \frac{t}{t-1} d_1 + d_1=:f(t)$. As $f$ is a monotone decreasing function, it is easy to see $f(t) \geq 9/4d_1$ only for $2 \leq t\leq 4$. As the function gives only an upper bound on \OPT's value and we have assumed $t\geq 3$ it suffices to show that actually $\OPT \leq 9/4d_1$ for the cases $t=3$ and $t=4$ in order to establish the statement.

If \OPT\ wants to gain more profit than twice the profit \NFDs gains, then it has to use at least three bins, which is because \NFDs fills the largest bin in the instance. It is clear our bound on \OPT\ gets only better, if there are less than three bins in the instance, i.\,e. $m<3$. We can assume \OPT\ uses exactly three bins, if we do not make any assumptions on their size, besides that they have at most the same demand as bin $1$ has.

For ease of notation we relabel the bins \OPT\ uses as $2$, $3$ and $4$ in non-increasing order of demand. It suffices to show that $d_2+d_3+d_4 \leq 9/4 d_1$. 

Clearly, if $d_2+d_3+d_4= c$ and $c$ is fixed, we can assume $d_2=d_3=d_4=c/3$, if we allow that \OPT\ may split the items $4,\dots,n$ arbitrarily. This can only be better for \OPT\ than any choice of $d_2$, $d_3$, $d_4$, since $\sum_{l=t+1}^n s_l < d_4$. Note, that the last bound holds, since we have assumed all bins besides bin $1$ are empty in \NFD's solution.

Now, for the sake of contradiction assume we have $d_2+d_3+d_4 > 9/4 d_1$. Because of $d_2=d_3=d_4=c/3$ we have that it must be $d_2,d_3, d_4 >3/4d_1$. For the case $t=3$ we have $s_1+s_2+s_3 \leq 3/2 d_1$ and we see that the items $1$, $2$, $3$ do not even suffice to fill the bins $d_2$ and $d_3$. As $\sum_{l=t+1}^n s_l < d_4$ we find that \OPT\ cannot fill all three bins and the claim is established. An analogous computation gives a contradiction for the case $t=4$, too.
\end{proof}

In order to simplify the following statements we introduce the term head of the instance, which is a distinguished bin. For this, fix a solution of \NFDs to a given instance $(I,J)$. Let $i_0$ be the index of the first not well-covered bin with $u(i_0)>0$ and let $i_1$ be the smallest index such that $u(i_1+1) = 0$ with $i_1 \geq i_0$. Let $i^* = \max_{i:u(i) > 2d_i} \{i \leq i_1\}$. Then, the bin $i^*$ is called the {\em head} (of the instance).

\begin{observation}\label{obs:headOfInstanceWellDef}
Fix a solution of \NFDs to the instance $(I,J)$. If there is a filled not well-covered bin, then $i^*$, the head of the instance, is well-defined.
\end{observation}
\begin{proof}
Recall that $i^* = \max_{i:u(i)>2d_i} \{ i \leq i_1 \}$. We verify that the set, over which the maximum is taken, is non-empty. By Observation~\ref{obs:wellCovWellDefined} we can assume there exists at least $k\geq1$ well-covered bins in every instance. Observe by definition that if there is an index $i$ with $u(i) >0$, such that $i$ is not a well-covered bin, all bins $i'$ with $u(i') > 0$ and $i' > i$ are also not well-covered. Let $i_0$ now be the smallest index of a bin with $u(i_0) > 0$, which is not well-covered, which exists by precondition. As $u(m) =0$ as guaranteed by Observation~\ref{obs:reduce2} there exists a smallest index $i_1$, with $i_1> i_0$ and $u(i_1+1) =  0$. By the definition of well-coverage there has to be an index $i$ with $i_0 \leq i < i_1$ such that $u(i) > 2d_{i}$. Now, let $i^* \leq i_1$ be the largest of such indices. Thus the indices $i_0$, $i_1$ and $i^*$ are well-defined and so is the term ``head'', if the solution of \NFDs contains a not well-covered, non-empty bin.
\end{proof}

\begin{lemma}\label{lem:onlyOneMachineBeforeGap2}
Let $(I,J)$ be an instance, for which \NFDs gives a solution with $k=1$ well-covered bin and there is a non-empty not well-covered bin. Let $i^*$ be the head of the instance. If there are at least three items on bin $1$ in \NFD's solution and \OPT\ assigns at least one of these items to a bin with index at least $i^*$, then $\OPT(I,J)/\NFD(I,J) \leq 9/4$.
\end{lemma}
\begin{proof}
Since there is only one well-covered bin in the instance, it is $u(2)=0$. Let $j'$ be the item on bin $i^*$ and recall that $s_{j'} > 2d_{i^*}$, since $i^*$ is the head of the instance. Since for every item $j$ from bin $1$ we have $s_j\geq s_{j'}$, we have for every such item $j$ that it will not only fill a bin with index $i'\geq i^*$, but even $d_{i'} \leq d_{i^*} \leq s_{j'}/2 \leq s_j/2$.

Since $t \geq 3$ we have that the item with largest index, which \NFDs assigns to bin 1, has size at most $d_1/2$. By the ordering of items by size we obviously have, if \OPT\ assigns at least one of the items from bin $1$ to a bin with index $i'\geq i^*$, we can assume it assigns -- besides possibly other items -- also the item $t$ to such a bin. This can only better than choosing an item with index smaller than $t$ by the fact that each such an item will fill its respective bin $i'$. For the remaining $t-1$ items on bin $1$ we can bound $\sum_{j=1}^{t-1}s_j < d_1$, because bin $1$ was not yet full, when item $t$ was assigned to it.

As in Lemma~\ref{lem:onlyOneMachineBeforeGap1} we can bound $\sum_{j=t+1}^n s_j < d_1$. Hence assigning at least one of the items to a bin with index $i' \geq i^*$ we can bound the profit \OPT\ yields with the above discussion by

$$\OPT < \sum_{j=1}^{t-1} s_j + s_t/2 + \sum_{j=t+1}^n s_j \leq d_1 + (d_1/2)/2 + d_1 = 9/4d_1,$$

and as $\NFD \geq d_1$, the claim follows.

\end{proof}

Now we give a straightforward upper bound for \NFDs in particular for the case there are $k=2$ or $k=3$ well-covered bins and at least one of them contains at least three items.

\begin{lemma}\label{lem:twoOrThreeMachinesBeforeGap}
Let $(I,J)$ be an instance such that \NFDs gives a solution with $k \geq 2$ well-covered bins. If at least one of these bins contains at least three items then $\OPT(I,J) / \NFD(I,J) \leq 9/4$.
\end{lemma}
\begin{proof} Note that the case $k\geq 4$ is already covered by Observation~\ref{obs:bigSolutionsGood} and hence we have only to prove the cases $k=2$ and $k=3$. Firstly consider the case $k=3$. Let \wlgs $1$, $2$ and $3$ be the bins, to which \NFDs assigned items. 

We may assume that if there are $t$ items altogether on the bins $1,2,3$ that there are at least $t$ empty bins of size at most $d_4$ in the instance, otherwise we give enough copies of bin $4$ to \OPT.

Recall, if on any bin $i=1,2,3$ reside three items in \NFD's solution we have $u(i) \leq 3/2 d_i$ and it is clear this term is smaller, if there are more than three items. Because of the ordering and the fact, that all bins $1$, $2$ and $3$ are well-covered we can bound $\sum_{l=1}^t s_l \leq 2d_1+2d_2+3/2d_3$. Further $\sum_{l=t+1}^n s_l < d_4\leq  d_3$ as bin $4$ must be empty, because there are three well-covered bins in the instance.

Hence we can bound the value of \OPT\ by the sum of item sizes in the instance $\OPT < 2d_1+2d_2+3/2d_3+d_3 =2d_1+2d_2+5/2d_3$. For \NFDs we yield $\NFD\geq d_1+d_2+d_3$ and hence 

$$\frac{\OPT}{\NFD} \leq 2+ \frac{1/2 d_3}{d_1+d_2+d_3} \leq 2+1/6 \leq 9/4.$$

For $k=2$ we deduce analogously $\OPT \leq 2d_1+3/2d_2+d_2= 2d_1+5/2d_2$ and $\NFD \geq d_1+d_2$. We yield $\OPT/\NFD \leq 9/4$. Note, that in particular it does not matter, if there are bins with index larger than $k+1$ that were covered by \NFD. It is easy to see that if there is another bin with $3$ items or on one bin there are more than three items, the bound is even better. Hence the claim follows.
\end{proof}

Note that Lemma~\ref{lem:twoOrThreeMachinesBeforeGap} uses a pure volume argument. Hence, regarding the preconditions, it does not matter, if there are additionally bins in the instance, which were filled by \NFD.

\begin{lemma} \label{lem:atMostTwoJobsOnIPrime}
Let $(I,J)$ be an instance such that \NFDs gives a solution with $k \in \{1,2,3\}$ well-covered bins, such that on each at most two items reside. Let the solution further contain at least one non-empty not well-covered bin, and let $i^*$ be the head of the instance. Define $I'=\{1,\dots,i^*\}$ and $I'' = I \setminus I'$. Further let $J'$ be the set of items in \NFD's solution, which reside on a bin from $I'$, and let $J'' = J\setminus J'$. Call the set of items \OPT\ assigns to a bin from $I'$ the set $A$ and let $B= J \setminus A$. Then $(\OPT(I',A)+\OPT(I'',B \setminus J''))/\NFD(I',J') \leq 2.$
\end{lemma}
\begin{proof}
At first we observe $\NFD(I',J') = \NFD(I',J'\cup A)$. This is because $A\setminus J'' \subseteq J'$ and by definition of $J''$ \NFDs does not assign any of the items from $J''$ to a bin from $I'$.

In case, $B \setminus J'' = \emptyset$, which means \OPT\ decides only to use the bins from $I'$, we can remove the bins in $I''$ from the instance and can compare the terms from the statement on the instance $(I',J' \cup A)$. This instance has filled not well-covered bins, but observe that these are now the right-most bins, i.\,e. there is an index $i'$, such that the bins $i',\dots, |I'|$ are not well-covered and $u(i) > 0$ for all $i' \leq i \leq |I'|$. Hence Observation~\ref{obs:reduce2} is applicable and the resulting instance consists only of well-covered bins. Then the claim follows from Lemma~\ref{lem:wellCoveredWithAtMostTwoJobs}.

Now assume $B \setminus J'' \neq \emptyset$. Because the smallest item of $J'$ resides alone on $i^*$ and $B \setminus J'' \subseteq J'$ we have every item from $B \setminus J''$ will reside alone on a bin from $I''$. By the same reason, when we add a set $I^*$ of bins with $|I^*| = |B \setminus J''|$, in which each bin has demand $d_{i^*}$, every newly introduced bin from $I^*$ will be covered by an item from $B \setminus J''$. Further the bins in $I^*$ are as least as large as the bins in $I''$. Hence we can move each item from $B \setminus J''$ to the bins in $I^*$ and remove the bins from $I''$, since this can only be better for \OPT.

We argue the proof of Observation~\ref{obs:reduce2} is still correct in this situation, even if \OPT\ has additional bins in the instance. Then we can assume that the bin with largest index from $I'$ is empty in \NFD's solution. The three things we have to keep an eye on are the following. Firstly we have to delete the same bins from the instances of \NFDs and \OPT. Further we can delete only the bins, which are covered by \NFD. At last, the bins to be deleted have to be the right-most ones in the instance, i.\,e. we may only delete the bins $i',\dots,m$, if $u(i) > 0$ for all $i\geq i'$. Then we may apply Observation~\ref{obs:reduce2} and have that the bin with largest index is empty in \NFD's solution.

In this situation we can show already with Lemma~\ref{lem:wellCoveredWithAtMostTwoJobs} that $\OPT(I'\cup I^*,A)/\NFD(I',A) \leq 2$. Firstly note it is justified to assume that no other item than the item we put on $i$ through the modification of the instance above is assigned to a bin $i$ from $I^*$. This is because \OPT\ has chosen this alternative to assign such an item to a bin from $I''$ and we actually gave \OPT\ at the most additional profit. If the additional bins in the instance of \OPT\ can take no other items than the items, which already reside on them, then the relaxation carried out in Lemma~\ref{lem:wellCoveredWithAtMostTwoJobs} is still possible, which we justify now. In this lemma we showed that even in a relaxed setting $\OPT$ could achieve no more than twice the profit $\NFD$ gained. Observe that in the relaxed setting it is only better for \OPT\ to assign the items now residing on the bins from $I^*$ to the respective bins from the $F_l\cup G_l$ sets. But then by Lemma~\ref{lem:wellCoveredWithAtMostTwoJobs} the claim follows.
\end{proof}

Now we have all tools at hand to decompose a given instance, such that our proof may restrict to analyze the specific parts of the decomposed instance.

\begin{lemma}[Decomposition Lemma]\label{lem:decompositionLemma}
Let $(I,J)$ be an instance, such that \NFDs gives a solution with $k$ well-covered bins and at least one not well-covered bin. Let $i^*$ be the head of the instance. Let $J'$ be the set of items residing on the bins $1,\dots,i^*$ in \NFD's solution and $J'' = J\setminus J'$. Further let $I' = \{1,\dots,i^*\}$ and $I'' = I \setminus I'$. Then $\OPT(I,J)/\NFD(I,J) \leq \max \{ 9/4,\OPT(I'',J'')/\NFD(I'',J'') \}.$
\end{lemma}  
\begin{proof}
At first we may assume that for the number of well-covered bins, $k$, we have $k\leq 3$, since otherwise the claim follows by Observation~\ref{obs:bigSolutionsGood}. 

Consider the case that all of the well-covered bins contain at most two items. Observe that $\NFD(I,J) = \NFD(I',J') + \NFD(I'',J'')$ holds for \NFDs by the definition of the sets. Consider the solution of \OPT. Let $A$ be the set of items, which reside on the bins in $I'$ and $B$ the set of items which reside on the bins in $I''$. Clearly, $\OPT(I,J) = \OPT(I',A) + \OPT(I'',B)$. 

We show that $\OPT(I'',B) \leq \OPT(I'',B \setminus J'') + \OPT(I'',J'')$. We want to emphasize, this is not a trivial relation, since it could be that \OPT\ could not use all items of the subsets of $B$, when we split this set up. But observe, here we have $B \setminus J'' \subseteq J'$. Thus all items in $B \setminus J''$ are as least as big as the smallest item in $J'$, which was the item on $i^*$. Hence for every item $j\in B \setminus J'' \subseteq J'$ we have $s_j \geq 2d_i$ for every bin $i \in I''$. Thus every such item resides alone on a bin in \OPT's solution for the instance $(I'', B)$ (and can also fill such a bin) and the inequality holds.

With this 
\begin{align}
\frac{\OPT(I,J)}{\NFD(I,J) }  & = \frac{\OPT(I',A) + \OPT(I'',B)}{ \NFD(I',J') + \NFD(I'',J'')}  \\[0.5\baselineskip] 
		             & \leq \frac{ \OPT(I',A) +  \OPT(I'',B \setminus J'') + \OPT(I'',J'')}{\NFD(I',J') + \NFD(I'',J'')} \\[0.5\baselineskip] 
			      \label{eq:elementary}  & \leq \max \left \{ \frac{ \OPT(I',A) +  \OPT(I'',B \setminus J'')  }{\NFD(I',J') } , \frac {\OPT(I'',J'')}{\NFD(I'',J'')} \right \} \\[0.5\baselineskip] 
			       & \label{allLemmas} \leq \max \left \{ \frac{9}{4} , \frac {\OPT(I'',J'')}{\NFD(I'',J'')} \right \}  ,
\end{align}
 where we used elementary calculus in~(\ref{eq:elementary}) and in~(\ref{allLemmas}) we have used Lemma~\ref{lem:atMostTwoJobsOnIPrime}.
 
Now consider the case that one of the well-covered bins contains at least three items. For $k\geq 2$ the claim follows by Lemma~\ref{lem:twoOrThreeMachinesBeforeGap} and we are left to show the claim for $k=1$.

We may restrict ourselves to instances, for which \OPT\ assigns all $t$ items, which reside on bin $1$ in \NFD's solution to a bin with index $i' \leq i^*$, as otherwise the claim follows by Lemma~\ref{lem:onlyOneMachineBeforeGap2}. Again we decompose the solutions of \NFDs and \OPT\ identically as above, where we are left to show step~(\ref{allLemmas}). By the assumption all items from bin $1$ are assigned to a bin with index $i'\leq i^*$ and as $B\setminus J'' \subseteq J'$ we have $B\setminus J'' =\emptyset$ and hence $\OPT(I'',B \setminus J'') = 0$. Thus it is enough to show $\OPT(I',A)/\NFD(I',J') \leq  9/4$. 

Again with the observation $\NFD(I',J') = \NFD(I',J' \cup A)$ and $\OPT(I',A) \leq  \OPT(I',J'\cup A)$ we can remove with Observation~\ref{obs:reduce2} the filled not well-covered bins and are left with an instance, which contains only one well-covered bin and all other bins are empty. The claim follows with Lemma~\ref{lem:onlyOneMachineBeforeGap1}.
\end{proof}


\begin{proof2}[of Theorem~\ref{thm:NFDisSuperb}] 
First observe that Example~\ref{exa:nfd_lower_bound} yields a lower bound of $9/4$ on the approximation ratio of \NFD. Let $k$ be the number of well-covered bins in the instance. If $k\geq4$ then Observation~\ref{obs:bigSolutionsGood} already gives the claim. Thus let $k \in \{ 1, 2, 3\}$. If all filled bins are also well-covered and one of these contains at least three items, then the claim follows from Lemma~\ref{lem:onlyOneMachineBeforeGap1} and Lemma~\ref{lem:twoOrThreeMachinesBeforeGap}. For the case that every one of the well-covered bins contains at most two items the statement follows from Lemma~\ref{lem:wellCoveredWithAtMostTwoJobs}.

Now let there be a filled but not well-covered bin in the instance. Define $I'=\{1,\dots,i^*\}$, $I'' = I \setminus I'$, $J'$ to be the set of items, which are assigned to the bins in $I'$ by \NFDs and $J''= J\setminus J'$, where $i^*$ is the head of the instance. Now we can apply Lemma~\ref{lem:decompositionLemma}.

Observe $(I'',J'')$ is a smaller instance, with at least one not well-covered bin less. Hence we can apply the analysis in a recursive step again on this instance. The recursion terminates if $(I'',J'')$ is an instance which has only well-covered bins or in which the solution of \NFDs has no covered bins. Clearly, in the latter case we have that \NFDs is optimal and in the former we can argue as above. The algorithm can be implemented such that the running-time is dominated by the sorting of the bins and items.
\end{proof2}

\subsubsection*{Monotonicity of Next Fit Decreasing for Variable-Sized Bin Covering}

For this subsection we introduce for sake of shortness the following notion. We will compare the solutions of \NFDs to some instances $(I,J)$ and $(I,J')$. For shortness we say in the instance $(I,J)$ a certain property holds, where we mean that the solution of \NFDs to the instance $(I,J)$ has this property.

\begin{prop}
\NFDs is a monotone algorithm for \prob{Variable-Sized Bin Covering}, i.\,e. if $(I,J')$ is an instance and $J\supseteq J'$, then it follows $\NFD(I,J) \geq \NFD(I,J')$. 
\end{prop}
\begin{proof}
Obviously it suffices to show the claim when the instance $(I,J)$ contains exactly one new item in comparison to the instance $(I,J')$, i.\,e. there is some $j\in J$ with $j\notin J'$ and $J=J'\cup\{j\}$. If all bins $i$, which are filled in the instance $(I,J')$, are filled as well in the instance $(I,J)$ then the claim follows. Thus assume there is a bin $i'$, which is filled in the instance $(I,J')$, but is not in the instance $(I,J)$. 

Since $J\supseteq J'$, it then has to be case that there is a bin $i$, which is covered in the instance $(I,J)$, but is not in the instance $(I,J')$. Moreover for the bin with smallest index of these, call it $i^*$, we have that $i^*<i'$. This is immediately by the behavior of \NFD, since otherwise \NFDs would have covered the bin $i^*$ in the instance $(I,J')$, too. Hence we have that for every bin $i=1,\dots,i^*-1$ that either $i$ is covered in both instances $(I,J)$ and $(I,J')$ or is not covered in both instances.

Let $j^*$ be the item with smallest index, which resides on a bin with index at least $i^*$ in the instance $(I,J')$ (or was not assigned). Since \NFDs did not cover the bin $i^*$ in the instance $(I,J')$, we have $d_{i^*} > \sum_{j=j^*}^n s_j$. This already gives the claim, as the bins $1,\dots,i^*-1$ were identically covered in both instances and in the instance $(I,J)$ additionally at least bin $i^*$ is covered, which yields profit $d_{i^*}$ and the profit, which can be gained on the bins $i^*,\dots,m$ in the instance $(I,J')$ is smaller as this, as shown.
\end{proof}

\subsection{Inapproximability in the Unit Supply Model}
By reduction from \prob{Partition} it is not hard to see that the classical \prob{Bin Covering} is \NP-hard and is not approximable within a factor of two, unless $\P = \NP$. This clearly extends to all of the models we consider here. Now the question arises if improvements in an asymptotic notion, where the optimal profit diverges, are possible. Note that we still require $p_i = d_i$, which yields that divergence of the optimal profit implies divergence of the total demand of the instance. However, it is not obvious how to define a suitable asymptotics in the \emph{unit supply} model: If only the total item size diverges, the optimal profit does not. If, in addition, the bin demands (but not their number) diverges, these instances still contain \prob{Partition}. Thus we consider an asymptotics, where the total item size, the total demand, and the number of bins diverges. The following theorem states that any algorithm can not have an approximation ratio of $2-\epsilon$, if $\epsilon>0$ is a constant, even in this case. Even stronger, as the choice of $s=\omega(m)$ is possible we have $\rho \to 2$, for $m\to \infty$.

\begin{theorem}
Consider \prob{Variable-Sized Bin Covering} with unit supply. Let $2\leq m \leq n$. Then there is an instance $(I,J)$ with $|J| = n+m-2$, for which an optimal algorithm covers $m$ bins, but there is no polynomial time algorithm with approximation factor better than $\rho=2-\frac{m-2}{s/2+m-2}$, unless $\P=\NP$.
\label{thm:lower_bound}
\end{theorem}
\begin{proof}
We use a reduction from the \prob{Partition} problem. Recall that for this we are given a set of items $P = \{1,\dots,n\}$, where item $j$ has integral size $s_j$. Our goal is to find an index set $L \subset P$, such that $s(L)= s(P\setminus L)$, i.\,e. the items from $P$ are partitioned in two sets of equal size.

Let $P'$ be a \prob{Partition} instance and we refer to the sizes of the items as $s_j'$. We define an instance $(I,J)$ for \prob{Variable-Sized Bin Covering}. We set $I=\{1,\dots,m\}$ and $J  =\{1,\dots,m+n-2\}$ with $s_j = 2s_j'm$ for $j=1,\dots,n$ and $s_j=1$ for $j=n+1,\dots,n+m-2$, i.\,e. the items of the \prob{Partition} instance are scaled by a factor of $2m$. As $m\leq n$ this is clearly done in polynomial time.

Recall, $s := \sum s_j$ and set $d_1 = s/2$, $d_2 = s/2$, where we assume $s/2$ is integral otherwise we output ``no'', which is due to the integral $s_j'$ (and thus integral $s_j$). Further we set $d_3 = \dots = d_m = 1$ and $d_i = p_i$ for all $i$.

Now we see that the solution of \prob{Variable-Sized Bin Covering} has a value of $s+m-2$, if the \prob{Partition} problem has a solution.

If the \prob{Partition} problem has no solution, we argue that the value of the solution to \prob{Variable-Sized Bin Covering} is at most $s/2+m-2$. Consider firstly the case that all items $1, \dots, n$ are assigned to bins $1$ and $2$ by an algorithm and the items $n+1,\dots,n+m-2$ are assigned to bins $3,\dots, m$. In the non-scaled instance $P'$, for every index set $L$ we have the property that 
$$s(L)\neq s(P\setminus L),$$

i.\,e. the left-hand sum and the right-hand sum differ by at least one, which is because the $s_j'$ were integral. Hence in the instance for \prob{Variable-Sized Bin Covering}, which uses the sizes from the scaled instance $P$ we have $u(1)$ and $u(2)$ differ for every assignment of the items $1,\dots,n$ to bins $1$ and $2$ by at least $2m$. Let w.\,l.\,o.\,g. be $u(1) > u(2)$, then we have $u(1) -  u(2) \geq 2m$, and thus $u(2) \leq s/2 -m$. Consequently, even, if all items $n+1,\dots,n+m-2$ are put by an algorithm on bin $2$, we have 

$$u(2) \leq s/2- m +(m-2) = s/2 -2 < s/2$$

and we see, bin $2$ is not covered, if the items $n+1,\dots,n+m-2$ are assigned arbitrarily. If the items $1,\dots n$ are assigned to arbitrary bins, the value of the solution may only decrease and we have shown that the value of a solution on the given instance is at most $s/2+m-2$, if the \prob{Partition} problem has no solution.  

Hence in case an algorithm has approximation ratio smaller than $\rho =2-\frac{m-2}{s/2+m-2}$, it can distinguish the cases and solve the \prob{Partition} problem. 
\end{proof}
\subsection{An A(F)PTAS for the Infinite Supply Model}
\label{sec:variable-sized_infinite_supply}

For the classical model Csirik, Johnson, and Kenyon~\cite{CsirikJohnsonKenyon:2001} were the first to give an APTAS. It turns out that the ideas of~\cite{CsirikJohnsonKenyon:2001} can be extended for the \prob{Variable-Sized Bin Covering} model with infinite supply. The basic idea is that small bin types can be ignored without losing too much profit. Then adjusting the parameters in the algorithm of~\cite{CsirikJohnsonKenyon:2001} and adapting the calculations gives the desired result. After that we can also adapt the method of Jansen and Solis-Oba~\cite{JansenSolis-oba:2003} to improve the running time and to obtain an AFPTAS. Here also the LP formulation has to be extended appropriately and subroutines have to be called with appropriately scaled parameters. We prove the following theorem in the next sections.
\begin{theorem}
There is an AFPTAS for \prob{Variable-Sized Bin Covering} in the infinite supply model.
\label{thm:afptas}
\end{theorem}

\subsubsection{An APTAS in the Infinite Supply Model}

It turns out that normalizing the demands of bins is advantageous here. Thus we assume in this section $1=d_1 > \dots > d_m>0$. Since we are in the \prob{Variable-Sized Bin Covering} model we have $d_i=p_i$ for all $i=1,\dots,m$. The result of this section will be the following.

\begin{theorem}\label{thm:aptas} There is an APTAS for the \prob{Variable-Sized Bin Covering} problem in the infinite supply model.
\end{theorem}

{\bf Outline of the APTAS.} Let $\epsilon >0$ be the desired approximation factor and we assume \wlgs that $1/\epsilon$ is integral. In the algorithm we delete all bin types with size at most $\epsilon$. This idea was also used by Murgolo~\cite{Murgolo:1987}. Then we subdivide the items into three sets: $L$ the set of large items, $M$ the set of medium-sized items, and $T$ the set of tiny items (for a formal definition of $L$, $M$, and $T$, see \ifthenelse{\boolean{full}}{Algorithm~\ref{alg:aptas}}{the full version}). The large items are further subdivided into $1/\epsilon^4$ groups, where each group has (almost) equal size with respect to the number of items it contains\ifthenelse{\boolean{full}}{ (cf. Step~\ref{aptas:roundDownLargeItems} of Algorithm~\ref{alg:aptas})}{}. This grouping technique originates from a paper of Fernandez de la Vega and Lueker~\cite{DeLaVegaLueker:1981}.

In each group, all items are rounded down to the size of the smallest item of the respective group. Note, this implies that there are at most $k=1/\epsilon^4$ many different sizes for the large items. The idea is here that the items of group $i$ can replace the items from group $i+1$ in an optimal solution. By this procedure only a profit bounded by the size of the first group is lost. Then all possible assignments -- referred to as configurations in the following -- of the large and rounded down items to bins are enumerated. 

Via an appropriate LP formulation a solution is determined. More precisely, an LP gives how many bins of each type are to be opened and according to which configuration a bin is assigned items. Here it is crucial that a configuration may not fill a bin, but the LP formulation ensures that only such sets of configurations in the solution are used so that the non-large items (i.\,e. the items from $M\cup T$) can fill the possible only partially covered bins in a greedy way. 
A listing of the algorithm can be found in Figure~\ref{fig:APTAS}.

{\bf Definitions for the APTAS.} We introduce a set of definitions, which will be helpful in order to write down the algorithm rather briefly. 

\begin{itemize}
\item Call a vector $v \in \{1,\dots,n\}^k$  a configuration, where $k=1/\epsilon^4$ is an integral constant. Let $\ell=(\ell_1,\dots,\ell_k)$ be the vector of large sizes, i.\,e. the size of the items in the respective group. The value $\ell_i$ is determined in Step~\ref{aptas:roundDownLargeItems} of the Algorithm~\ref{alg:aptas}. 

\item Let $e_j \in \{0,1\}^k$ the vector with an entry $1$ at position $j$ and $0$ at all positions $j'\neq j$.

\item For a configuration $v_i$ let $\ell(i) := v_i^\top \cdot \ell$, i.\,e. be the sum of sizes of all items, which are contained in $v_i$. Let $n(i,j) = v_i^\top \cdot e_j$ be the number of items of size $\ell_j$ in configuration $v_i$.

\item Let $C_j = \{ i \in \{1,\dots,r\} \mid d_{j-1} > \ell(i) \geq d_j \}$, where $d_0 = \infty$, i.\,e. associate every configuration to a bin type of largest size, such that the configuration covers the bin and let $C_j$ be the set of indices such that $C_j$ contains all the configurations associated to bin type $j$. 

\item Let $r(i,j) = d_j-\ell(i)$ be the demand of bin of type $j$, which is not covered by the configuration $i$ (the remainder).

\item Let $\tilde{C_j} = \{ i \in \{ 1,\dots,r\} \mid  r(i,j) > 0\}$, i.\,e. in $\tilde{C_j}$ there all configurations, which do not cover bin type $j$.

\item Let $n(i)$ be the number of items of size $\ell_i$, for $1\leq i\leq k$, in the instance. 

\end{itemize}

For the analysis let $\OPT(I;L,T)$ denote the value of an optimal algorithm, which assigns the items of $L$ and $T$ according to the bin types in $I$, where it may split the items in $T$ arbitrarily.

Now, we give the key observation, which lets us adapt the algorithm from~\cite{CsirikJohnsonKenyon:2001} to the \prob{Variable-Sized bin covering} model with infinite supply of bins.

\begin{observation} \label{obs:ThrowAwaySmallBins}
Fix an instance $(I,J)$. Let $I' := \{ i \in I \mid d_i > \epsilon\}$ be the set of bins, which have demand more than $\epsilon$. Then $(1+\epsilon)\OPT(I',J)+1 \geq \OPT(I,J)$. 
\end{observation}

\begin{proof}
Consider an optimal solution, in which only bins from the set $I\setminus I'$ are covered, since the bound is even better otherwise. We partition the items residing on bins in $I\setminus I'$ in two sets $J_1$ and $J_2$. In the set $J_1$ every item has size smaller than $\epsilon$, in set $J_2$ every item has size at least $\epsilon$. Observe that in an optimal solution the items from $J_1$ and $J_2$ will reside on distinct bins. Hence it suffices to show that for every part $i=1,2$ of the instance we have $(1+\epsilon)\OPT(I',J_i) \geq \OPT(I,J_i)$ and we may lose one additional bin in both instances together.

Clearly, if we put all items from $J_1$ with \NFDs on bins of size $1$, then for every filled bin $i$ we have $u(i) < 1+\epsilon$. Hence if $s(J_1)$ is the overall size of all items in $J_1$ we yield profit at least $\lfloor s(J_1)/(1+\epsilon)\rfloor \geq  s(J_1)/(1+\epsilon)-1$, whilst $s(J_1)$ is an upper bound for the profit \OPT\ yields on this part of the instance.

The items from $J_2$ we also assign with \NFDs to bins of size $1$. If the last bin is not filled and so was the last bin with the items from $J_1$, then we assign the items from this bin to the bin with the items from $J_1$. Hence there is at most one not filled bin and it suffices to show, we yield a profit of at least a $1/(1+\epsilon)$ fraction of the profit \OPT\ yields on the instance $(I,J_2)$.

\begin{figure}\label{alg:aptas}
\scalebox{0.96}{
\framebox[16.5cm]{\begin{minipage}{\textwidth}
\begin{enumerate}[{Step }1.]
\item Remove all bin types with size smaller than $\epsilon$ and let $m$ now be the number of remaining bin types in the instance.\label{step:pruneBins}

\item If $n < \lceil s/\epsilon^3 \rceil +\lfloor s/\epsilon \rfloor$, then set $L=\{1,\dots,n\}$ as set of large items and $M=T = \emptyset$ as sets of medium and tiny items. Else order items non-increasingly and define the sets $L$, $M$, $T$ of large, medium and tiny items as $L = \{1,\dots,  \lceil s /  \epsilon^3 \rceil \}$, $M= \{ |L|+1, \dots,   |L|+\lfloor s/\epsilon \rfloor \}$ and $T = \{ |L|+|M|+1,\dots, n\}$.

\item \label{aptas:roundDownLargeItems} Subdivide the items of $L$ in $k=1/\epsilon^4$ groups. Let $p = |L| \text{ div } k$ and $q= |L| \mod k$, then the groups $1,\dots,q$ have $p+1$ items each and the groups $q+1,\dots,k$ have $p$ items each. In every group $i$, we round down the size of every item of that group to the size $s_j =: \ell_i$, where $s_j$ is the size of the smallest item in the group.

\item Enumerate all configurations, such that we have $v_1,\dots,v_r$ are all the configurations and with this compute $C_j$ and $\tilde{C}_j$ for $j=1,\dots,m$.

\item Introduce variables $y_i$ and $z_{i,j}$ such that for $1 \leq i \leq r$ the variable $y_i$ is associated to configuration $i$ and the variable $z_{i,j}$, with $1\leq i \leq r, 1\leq j\leq m$, is associated to configuration $i$ and bin type $j$.

\item Compute $s(T)$ and $n(i)$ for $i=1\dots,k$.

 \label{aptas:lp}
Solve the following LP
  \begin{alignat}{2}
      \label{lp:aptas} & \text{maximize}   & \qquad \sum_{j=1}^m d_j (\sum_{i \in C_j} y_i + \sum_{i \in \widetilde{C}_j} z_{i,j}) & \\ 
      \notag & \text{subject to} & \sum_{i=1}^r n(i,j) \left ( y_i  + \sum_{l=1}^m  z_{i,l} \right) & \leq  n(j) \qquad j \in \{1,\dots,k\} \\
       \notag&                  & \sum_{j=1}^m \sum_{i \in \widetilde{C}_j} r(i,j)z_{i,j}          &\leq  s(T)                        \\
        \notag&                  & y_i                                                              &\geq  0     \qquad i \in \{1,\dots,r\} \\
        \notag&                  & z_{i,j}                                                          &\geq  0     \qquad i \in \{1,\dots,r\}, j \in \{1,\dots,m\}
\end{alignat}

\item Set for every variable of the LP $y_j' := \lfloor y_j \rfloor$  and $z_{i,j}' = \lfloor z_{i,j} \rfloor$. \label{aptas:roundDown}
\item Construct a solution in the following way. 
  \begin{enumerate}
    \item For every configuration $j=1,\dots,r$ take $y_j'$ many bins of the associated unique type and fill every bin accordingly to the configuration $v_i$. 
    \item For every pair $(i,j)$ take $z_{i,j}'$ many bins with demand $d_j$ and assign items accordingly to configuration $v_i$. 
    \item Fill the bins created accordingly to the $z_{i,j}'$ variables in a greedy way -- for example with \NFDs -- using the items $\lfloor s/\epsilon^2 \rfloor+1, \dots, n$, where we are left to show, that this is possible.
   \end{enumerate}
  \label{aptas:constructSolution}
\end{enumerate}
      \end{minipage}}}
      \caption{The APTAS.}
      \label{fig:APTAS}
\end{figure}

Consider a bin $i$, which is covered by $t$ many items from $J_2$ and let $S_i$ be this set of items. If it has fill level $u(i) \leq 1+\epsilon$ the claim follows, thus assume $u(i) > 1+\epsilon$ and let $\epsilon' = u(i)-1$. Let $I_i$ be the set of bins to which \OPT\ assigned the items from $S_i$ and recall that $|I_i|=|S_i|=t$, i.\,e. every item resides alone on its bin in the solution of \OPT. Since we have $\epsilon' > \epsilon$ and the last item was a smallest on bin $i$, we have for all $j \in I_i$ that $u_\OPT(j) -d_j \geq \epsilon'-\epsilon$, that is also \OPT\ wasted at least a volume of $\epsilon'-\epsilon$ per item from $I_i$ when it had assigned the items from $S_i$ to their respective bins in $I_i$. Since $|I_i|=t$ it follows that the profit \OPT\ gains for every such a bin $i$, which \NFDs fills with $u(i)=1+\epsilon'>1+\epsilon$, is bounded by
$$  1+\epsilon'-t(\epsilon'-\epsilon)  = 1- \epsilon'(t-1) + \epsilon t \leq 1- \epsilon (t-1) + \epsilon t =  1+\epsilon.$$
The profit \NFDs yields is at least $1$, hence, the claim follows.
\end{proof}

\begin{observation}\label{obs:LB4Opt}
We have $s \leq 2\OPT(I, L\cup T \cup M)+2$.
\end{observation}
\begin{proof}
We can assume \wlgs that the largest items in the instance have size less than $1$, since otherwise a preprocessing can remove larger items and assign them to the bin type with size $1$, which is clearly optimal. Then it is easy to see that $\OPT(I, L\cup T \cup M) \geq \lfloor s/2 \rfloor$, since already \NFDs gives such a bound using only the largest bin type with demand $1$. Rearranging and taking into account the rounding gives the claim.
\end{proof}

\begin{observation}\label{obs:noMediums}
Let $s\geq2$ and $\epsilon \leq 1/6$. Then $\OPT(I, L \cup M \cup T) \leq \OPT(I, L \cup T)/(1-2\epsilon)+2$.
\end{observation}
\begin{proof}
Take an optimal covering of the bins with all items, i.\,e. with items from $L\cup T \cup M$. Since \OPT\ yields at most $\lfloor s/\epsilon \rfloor$ many bins the average number of large items per bin is at least $1/\epsilon^2$ by the definition of the set $L$.

Hence removing $\lfloor s \epsilon \rfloor +1$ bins with the largest items, removes at least $\lfloor s/\epsilon \rfloor$ many large items. These can now be used instead of the medium-sized items in the rest of the instance, since there are at most so many medium items in the instance. The modified solution has at most $\lfloor s/\epsilon \rfloor+1$ bins less than the optimal solution.

With $\OPT(I, L\cup T \cup M) \geq \lfloor s/2 \rfloor \geq s/2-1$ as argued in the proof of Observation~\ref{obs:LB4Opt} and as $\epsilon \leq 1/6$ the removal of the $\lfloor s \epsilon \rfloor+1$ bins is possible, since an optimal solution contains at least this many bins. Further we have shown that $\OPT(I, L\cup M \cup T) - \lfloor s\epsilon \rfloor - 1 \leq \OPT(I, L\cup T)$.

With this,

\begin{align*}
 \OPT(I, L\cup T) & \geq \OPT(I, L\cup M \cup T)  - \lfloor s\epsilon \rfloor -1 \\
    & \geq \OPT(I, L\cup M \cup T)  - s\epsilon -1 \\
    & \geq \OPT(I, L\cup M \cup T) -  2\epsilon\OPT(I, L\cup T \cup M)-2\epsilon-1 \\
    & \geq (1-2\epsilon)  \OPT(I, L\cup M \cup T) -2\epsilon-1,
\end{align*}

where we have used Observation~\ref{obs:LB4Opt} in the third line. Rearranging and using $\epsilon \leq 1/6$, the claim follows.

\end{proof}

\begin{observation}\label{obs:OPTandFluid}
If $\epsilon \leq 1/6$ then $\OPT(I,L\cup M \cup T) \leq \frac{1+\epsilon}{1-2\epsilon} \OPT(I';L,T)+4$.
\end{observation}
\begin{proof}
We have 
\begin{align}
\OPT(I,L\cup M \cup T) & \leq (1+\epsilon)\OPT(I',L \cup M \cup T)+1 \label{eq:1} \\ 
			& \leq (1+\epsilon)  \left (\frac{1}{1-2\epsilon} \OPT(I',L\cup T)+2 \right)+1 \label{eq:2} \\
			& \leq \frac{1+\epsilon}{1-2\epsilon} \OPT(I';L,T) +4, \label{eq:3}
\end{align}
where~(\ref{eq:1}) is by Observation~\ref{obs:ThrowAwaySmallBins},~(\ref{eq:2}) is by Observation~\ref{obs:noMediums}, and~(\ref{eq:3}) is by the precondition $\epsilon \leq 1/6$ and the observation $\OPT(I',L\cup T) \leq \OPT(I';L,T)$.
\end{proof}

\begin{observation}\label{obs:roundingLargeItem}
Let $\epsilon \leq 1/10$. Consider the solution according to $\OPT(I';L,T)$. If $L'$ denotes the sets of large items, which is obtained by rounding down the item sizes from $L$ as done in Step~\ref{aptas:roundDownLargeItems} of Algorithm~\ref{alg:aptas}, then
$$ \OPT(I';L,T) \leq \frac{1-2\epsilon}{1-4\epsilon-2\epsilon^2}  \OPT(I';L',T) +9.$$
    
\end{observation}
\begin{proof}
By the rounding procedure in Step~\ref{aptas:roundDownLargeItems} we have that in a solution, in which the items from $L$ were rounded down, an item from a group $i$ can replace an item from the group $i+1$, where we lose the $p+1$ largest items or $p$ largest items, if $|L|$ can be divided by $k$. Hence, if $L'$ denotes the set of rounded large items, we have $\OPT(I';L,T) \leq \OPT(I'; L',T)+p+1$.

With this, since $p \leq \lceil s/\epsilon^3 \rceil \cdot \epsilon^4$ by the number of groups, we have $\OPT(I';L,T) \leq \OPT(I'; L',T) +s\epsilon +2$. Hence, we can bound

\begin{align}
 \OPT(I';L,T) & \leq \OPT(I'; L',T) +s\epsilon +2 \\
	     & \leq \OPT(I'; L',T) + 2\epsilon\OPT(I,L\cup M \cup T) +2\epsilon +2 \label{eq:a1} \\
	     & \leq \OPT(I'; L',T) + \frac{1+\epsilon}{ \frac{1}{2\epsilon}-1} \OPT(I';L,T)+10\epsilon +2\label{eq:a2} \\
	     & \leq  \OPT(I'; L',T) + \frac{1+\epsilon}{ \frac{1}{2\epsilon}-1}  \OPT(I';L,T) +3,\label{eq:a3}
\end{align}
where we used Observation~\ref{obs:LB4Opt} in~(\ref{eq:a1}), Observation~\ref{obs:OPTandFluid} in~(\ref{eq:a2}) and the fact that $\epsilon \leq 1/10$ in step~(\ref{eq:a3}). Rearranging terms and using again $\epsilon \leq 1/10$ gives the claim.
\end{proof}

\begin{proof}[of Theorem~\ref{thm:aptas}]
Assume $\epsilon\leq 1/10$. We firstly show that all bins the algorithm outputs are filled, then we bound the approximation ratio.

As argued in Observation~\ref{obs:roundingLargeItem}, we can obtain a solution for the original problem from the solution with the large items, which were rounded down in Step~\ref{aptas:roundDownLargeItems} of Algorithm~\ref{alg:aptas}. Hence it suffices to show that in Step~\ref{aptas:roundDown} all opened bins are filled.

The bins in which items from configurations associated to variables $y_j'$ are filled by definition of the $y_j$ variables. For bins corresponding to the $z_{i,j}'$ variables we argue that these are filled in Step~\ref{aptas:roundDown}\,(c). By the size of the smallest bins any solution has at most $\lfloor s/\epsilon \rfloor$ many bins. Since \NFDs does not assign items to filled bins, we have that at most $\lfloor s/\epsilon  \rfloor$ items assigned by \NFDs do not cover any demand and we say the sum of item sizes is wasted. 

Since \NFDs wastes at the worst the largest $\lfloor s/\epsilon \rfloor$ of the items in $M\cup T$, we have that at most a sum of sizes of $s(M)$ is wasted, by definition of the set $M$. Since $s(T)$ is at most the demand to cover as enforced by the constraints of the LP, and $s(M \cup T)$ is the sum of the items sizes, which is available to \NFDs in order to fill the not covered bins induced by the $z_{i,j}'$ variables, we have \NFDs can fill all those bins.

We now give the calculation for the approximation factor, where we explain the steps thereafter. Let $I' = \{ i \in I \mid d_i > \epsilon\}$. We have

\begin{align}
 \OPT(I, L\cup M \cup T)  & \leq \frac{1+\epsilon}{1-2\epsilon} \OPT(I';L,T)+4  \label{cor:s20} \\
		      & \leq   \frac{1+\epsilon}{1-2\epsilon} \left (  \frac{1-2\epsilon}{1-4\epsilon-2\epsilon^2} \OPT(I';L',T) +9  \right)+4  \label{cor:s21} \\
		      & \leq   \frac{1+\epsilon}{1-4\epsilon-2\epsilon^2}  \OPT(I';L',T) +17  \label{cor:s22} \\
		     & \leq   \frac{1+\epsilon}{1-4\epsilon-2\epsilon^2}  \left (\sum_{i=1}^m   d_i \left ( \sum_{c\in C_i} y_c  + \sum_{c\in \widetilde{C}_i} z_{i,c} \right) \right  )      + 17  \label{cor:s6} \\
		     & \leq  \frac{1+\epsilon}{1-4\epsilon-2\epsilon^2}  \left (\sum_{i=1}^m   d_i \left ( \sum_{c\in C_i} y_c'  + \sum_{c\in \widetilde{C}_i} z_{i,c}' \right)   + 1/\epsilon^4    \right )    + 19  \label{cor:s7}  \\
\end{align}

In~(\ref{cor:s20}) we use Observation~\ref{obs:OPTandFluid}. In~(\ref{cor:s21}) we apply Observation~\ref{obs:roundingLargeItem}. (\ref{cor:s22}) uses the fact that $\epsilon \leq 1/10$.~(\ref{cor:s6}) is easy to observe and finally, in~(\ref{cor:s7}) we round down the variables of the LP, which is explained as follows.

We have that our LP has only $1+1/\epsilon^4$ constraints besides the non-negativity constraints. Hence an optimal basic solution has at most $1+1/\epsilon^4$ fractional values and hence we lose at most so many bins with demand $1$ due to rounding down the fractional variables. 

As such a solution can be found in polynomial time, since the LP has polynomial size in $n$ (though exponentially in $1/\epsilon$, which is a constant) the bound follows. If $1+\epsilon'>1$ is the desired approximation ratio we set $\epsilon =  \epsilon'/13$ and run our algorithm which gives an approximation ratio of at least $1+\epsilon'$ minus a constant term. Also observe, by our choice of $\epsilon$, and $\epsilon'\leq 1$ \wlg, the assumption $\epsilon \leq 1/10$ was justified.
\end{proof}

\subsubsection{An AFPTAS in the infinite supply model}

Jansen and Solis-Oba~\cite{JansenSolis-oba:2003} gave an AFPTAS for the \ALG[Bin Covering] problem. In this section we extend their method to work for \ALG[Variable-sized Bin Covering] in the infinite supply model in order to prove Theorem~\ref{thm:afptas}.

\paragraph{Formulation as a resource sharing problem and overall method.}

The AFPTAS does not solve the linear program~(\ref{lp:aptas}) (LP) in Step~\ref{aptas:lp} of the APTAS exactly. Instead we approximate LP~(\ref{lp:aptas}). We will show later, how to transform this solution into a feasible solution for LP~(\ref{lp:aptas}). Then we apply the rounding procedure form Theorem~\ref{thm:aptas}. Recall that $k=1/\eps^4$ is the number of different large sizes and $r= n^k=n^{1/\eps^4}$ is the number of configurations. Let $x =(y_1,\dots,y_r,z_{1,1},\dots,z_{1,m},z_{2,1},\dots,z_{2,m}, \dots, z_{r,1},\dots,z_{r,m})$ be a solution vector to the \ALG[Variable-Sized Bin Covering] problem. We restate LP~(\ref{lp:aptas}) in the following form.

 \begin{alignat}{2}
\label{lp:afptas}	  \lambda^* = \min \bigl \{ \lambda \mid     \bigr.   \\
         \notag&  \sum_{i=1}^r \frac{ n(i,j)}{n(j)} \left ( y_i  + \sum_{l=1}^m  z_{i,l} \right)  \leq  \lambda & \qquad  1 \leq j \leq k , x \in B_t \\
  \notag&    \sum_{j=1}^m \sum_{i \in \widetilde{C}_j} \frac{r(i,j)}{s(T)}z_{i,j}     \leq  \lambda &  x \in B_t  \\
                      \notag  \bigl. \bigr  \},
\end{alignat}
where

\begin{alignat}{2}
\notag    B_t= \bigl \{ x \bigr . \mid     &  \sum_{j=1}^m  d_j \left ( \sum_{i\in C_j} y_i  + \sum_{l=1}^m \sum_{i\in \widetilde{C}_j} z_{i,l} \right ) = t \eps \bigl . \text{ and } \forall i,l: y_i\geq 0, z_{i,l} \geq 0 \bigr \}
\end{alignat}

Note that for $\lambda =1$ the constraints of LP~(\ref{lp:aptas}) are equivalent to the constraints of LP~(\ref{lp:afptas}). The value $t$ defining the simplex $B_t$ thereby will be set such that $t\eps$ is the (approximate) value of an optimal solution and we can guess $t$ via binary search. We explain this in more detail later. Suppose $t\eps$ is the true value of an optimal integral solution. Then $\lambda^*=1$ is the optimal value of the resource sharing problem and the corresponding solution vector $x$ gives also a solution to LP~(\ref{lp:aptas}). Jansen and Solis-Oba give in~\cite{JansenSolis-oba:2003} a solution to LP~(\ref{lp:afptas}) for the case when $m=1$ and $d_1=1$ with the {\em price directive decomposition method}~\cite{GrigoriadisKhachiyan:1996,JansenZhang:2002} and show how this can be transformed into a $(1+\eps)$-approximation for the \prob{Variable-Sized Bin Covering} problem. We can extend their technique to work for $m$ bin types.

LP~(\ref{lp:afptas}) is a {\em convex block-angular resource sharing problem}. Resource sharing problems can be solved with the price-directive decomposition method~\cite{GrigoriadisKhachiyan:1996,JansenZhang:2002} within any given approximation factor. We give a brief overview of this method. 

An algorithm for solving a resource sharing problem finds a solution iteratively. It starts with an arbitrary feasible solution $x^*$ and determines a price vector $p=(p_1,\dots,p_{k+1})$, whose components are non-negative. It requires to solve a subproblem, called the {\em block program}, whose solution depends on $p$. We will state the block program for our problem below. A linear combination of an optimal solution $\hat{x}$ to the block program and the previous solution $x^*$ found by the price-directive decomposition method so far determines an updated solution $x^*$ for the original resource sharing problem. After a certain number of iterations for any given $\delta>0$ the price-directive decomposition method guarantees a solution $x^*$ with objective value at most $(1+\delta)\lambda^*$, where $\lambda^*$ is the objective value of an optimal solution for the resource sharing problem. We are left to show how to transform this solution into a $(1+\eps)$-approximate solution for the \ALG[Variable-sized Bin Covering] problem.

\paragraph{Statement of the block program.}

The block program we have to solve is

\begin{align}
 \min \{ p^T A x \mid x \in B_t \},  \label{blockproblem}
\end{align}

where matrix $A=(a_{j,i})$ denotes the $(k+1)\times r(1+m)$ coefficient matrix corresponding to the constraints of LP~(\ref{lp:afptas}). We give the entries of $A$ in more detail. Let $f(j) = (j-r) \div m$ and $g(j) = (j-r) \mod m$. Then 

$$ a_{j,i} = \begin{cases}  \vspace{4pt}
 \dfrac{ n(i,j)}{n(j)}& \text{ if $1\leq j \leq k$ and $1 \leq i \leq r$} \\
\dfrac{ n(i',j)}{n(j)}  & \text{ if $1\leq j \leq k$ and $r+1 \leq i \leq r(1+m)$, where $i' = f(i-1)+1$} \\
0 & \text{ if $j= k+1$ and $1 \leq i\leq r$} \\
  \dfrac{r(i',j')}{s(T)} & \text{ if $j = k+1$ and $r+1 \leq i \leq r(1+m)$, where $i' = f(i-1)+1$ and $j' = g(i-1)+1$}. 
 \end{cases}$$

Observe, the coefficients in columns $1\leq i\leq r$ are the coefficients of the $y_i$ variables and the coefficients in columns $r+1 \leq j \leq r(1+m)$ are the coefficients for the variables $z_{i',j'}$, with $i' = f(i-1)+1$ and $j' = g(i-1)+1$. Note that we neglected here for ease of presentation that some additional entries $a_{j,i}$ for $j=k+1$ and $r+1 \leq i \leq r(1+m)$ maybe zero, namely if $i\notin \tilde{C}_j$, i.\,e. if configuration $i$ covers bin type $j$.

Since $B_t$ is a simplex, an optimal solution $x^*$ for program~(\ref{blockproblem}) will be attained at a vertex. That is one component of $x^*$ has value $t\eps$ and all other components are zero. Hence an optimal solution corresponds to a single configuration. Thus it is enough to find a configuration with smallest price, in order to solve the block problem, where the price of a configuration is determined by $p^TA$ as follows. 

Let $1\leq i \leq r$ be a configuration. Recall, that we have $m+1$ variables for configuration $i$, which are the variables $y_i,z_{i,1},\dots,z_{i,m}$. Each of these variables was associated to a bin type $j$. The variable $z_{i,j}$ for $1 \leq j \leq m$ was associated to bin type $j$ and the variable $y_i$ was associated to a bin type with largest demand, which is covered by $i$. 

Fix now one of these $m+1$ variables and say $j'$ is the index of this variable in the solution vector $x=(y_1,\dots,y_r,z_{1,1},\dots,z_{1,m},z_{2,1},\dots,z_{2,m}, \dots, z_{r,1},\dots,z_{r,m})$. Note that by fixing a variable also a bin type $j$ is fixed and also the other way round. We define the price of configuration $i$ with respect to bin type $j$ as $p^TAe_{j'}$, where $e_{j'}$ again denotes the vector with a one in row $j'$ and zero otherwise. That is, the price of configuration $i$ with respect to bin type $j$ is determined by multiplying the price vector $p$ with column $j'$ in the matrix $A$. Let $p=(p_1,\dots,p_{k+1})$. Then the price of a configuration $i$ with respect to bin type $j$ is 
$\sum_{l=1}^k n(i,l)p_l /n(l)$, if configuration $i$ covers a bin of type $j$ or $\sum_{l=1}^k n(i,l)p_l /n(l) + r(i,j)p_{k+1}/s(T)$ if it does not cover a bin of type $j$.

As argued it is enough to find a configuration $i$ with smallest price with respect to some bin type $j$. Define two types of integer programs (IP)
 \begin{align}
  \tau_{i,1} & = \min\sum_{l=1}^k \frac{p_l}{n(l)}u_l  \label{ip:findCoveringConfig} \\
   \text{s.\,t.} \quad & \sum_{l=1}^k \ell_l u_l \geq d_i \notag \\
    & u_l \in \{ 0,\dots,n(l)\} \notag
 \end{align}
 
 and
 
 \begin{align}
  \tau_{i,2} & = \min\sum_{l=1}^k  \frac{p_l}{n(l)}u_l + p_{k+1}\frac{d_i-\sum_{l=1}^k\ell_lu_l}{s(T)} \label{ip:findnonCoveringConfig}  \\
   \text{s.\,t.} \quad & \sum_{l=1}^k \ell_l u_l \leq d_i. \notag \\
   & u_l \in \{ 0,\dots,n(l)\} \notag
 \end{align}
 
 Here the variables $u_l$ denote the number of items of size type $l$ to choose. Hence it is not hard to see that IP~(\ref{ip:findCoveringConfig}) finds a cheapest configuration, which covers bin type $j$ and IP~(\ref{ip:findnonCoveringConfig}) finds a cheapest configuration, which does not cover a bin type $j$. Hence taking the overall cheapest configuration, i.\,e. the configuration which gives the minimum value in the set $\mathcal{M}= \{ \tau_{j,1},\tau_{j,2}\mid 1 \leq j\leq m\}$, is the configuration which is the solution to the block problem.
 
\paragraph{Solution of the block problem.}
 In the previous section we reduced the problem of finding an optimal solution to the block problem to finding the configuration, which gives the minimum value of the set $\mathcal{M}$ and in this section we show, how to find it by solving IPs~(\ref{ip:findCoveringConfig}) and~(\ref{ip:findnonCoveringConfig}).
 
 IP~(\ref{ip:findCoveringConfig}) is the minimum knapsack problem and it is folklore that there exists a FPTAS for it. By a dynamic program and an appropriate rounding technique Jansen and Solis-Oba~\cite{JansenSolis-oba:2003} can also obtain an FPTAS for the program $\min \{ \tau_{1,1}, \tau_{1,2}\}$, where $d_1 =1$. We can use their FPTAS as a procedure in order to find the overall cheapest configuration, i.\,e. for $m$ different and arbitrary $d_i$ values. For this we scale the constraints appropriately:

Let $\ell=(\ell_1,\dots,\ell_k)$ be the vector of all item sizes and $p=(p_1,\dots,p_{k+1})$ the price vector. Let $\ell(d_i) = (\ell_1/d_i,\dots,\ell_k/d_i)$ and $p(d_i)=(p_1,\dots,p_k,d_ip_{k+1})$. We replace the coefficients in the constraints of IPs~(\ref{ip:findCoveringConfig}) and~(\ref{ip:findnonCoveringConfig}) by the corresponding coefficients from the vectors $p(d_i)$ and $\ell(d_i)$. We observe that $u=(u_1,\dots,u_k)$ is a solution to the program
 
\begin{align}
  \tau_{i,1}' & = \min\sum_{l=1}^k \frac{p_l}{n(l)}u_l  \label{ip:findCoveringConfig2} \\
   \text{s.\,t.} \quad & \sum_{l=1}^k \frac{\ell_l}{d_i} u_l \geq 1, \notag \\
   & u_l \in \{ 0,\dots,n(l)\} \notag
 \end{align}
 
if and only if $u$ is a solution to IP~(\ref{ip:findCoveringConfig}). Note further that the respective objective values of the solutions are identical in both problems, since all scaled values do not contribute to the objective function. Hence a solution $u$ of IP~(\ref{ip:findCoveringConfig}) with objective value $\tau_{i,1}'$ is a solution of $u$ of IP~(\ref{ip:findCoveringConfig2}) with identical objective value.

Similarly we conclude that the program
 
  \begin{align}
  \tau_{i,2}' & = \min\sum_{l=1}^k  \frac{p_l}{n(l)}u_l + d_ip_{k+1}\frac{1-\sum_{l=1}^k u_l\ell_l/d_i}{s(T)} \label{ip:findnonCoveringConfig2}  \\
   \text{s.\,t.} \quad & \sum_{l=1}^k \frac{\ell_l}{d_i} u_l \leq 1. \notag \\
   & u_l \in \{ 0,\dots,n(l)\} \notag
 \end{align}

has a solution $u$ with objective value $\tau_{i,2}'$ if and only if $u$ is a solution to IP~(\ref{ip:findnonCoveringConfig}) with the same objective value. Note that IP~(\ref{ip:findCoveringConfig2}) and IP~(\ref{ip:findnonCoveringConfig2}) are of the shape of IPs~(\ref{ip:findCoveringConfig}) and~(\ref{ip:findnonCoveringConfig}), where $d_i=1$. Hence the FPTAS of Jansen and Solis-Oba for the block problem is applicable in this setting. Overall we have found an algorithm for solving the problem $\min \{ \tau_{i,1}, \tau_{i,2}\}$: divide the item sizes in $\ell$ by $d_i$ and multiply the $(k+1)$-st component of the price vector $p$ with $d_i$ and compute a solution of the block problem with the modified size and price vector with the FPTAS of Jansen and Solis-Oba in~\cite{JansenSolis-oba:2003}.

As argued the configuration minimizing $\min \{ \tau_{i,1}, \tau_{i,2}\}$ over all bin types $i=1,\dots,m$ is a $(1+\eps)$-approximate solution for the block problem of \prob{Variable-Sized Bin Covering}.

\paragraph{Approximation guarantee and running time analysis.}

\begin{lemma}A solution to LP~(\ref{lp:aptas}) with objective value at least $(1-2\eps)\OPT-O(1/\eps^4)$ and length $O(1/\eps^4)$ can be found in polynomial time.
\label{lem:approxLPAPTAS}
\end{lemma}

\begin{proof} As in the proof Theorem~\ref{thm:aptas} we assume \wlgs that $d_1=1$ and that the size of each item is smaller than $1$. Then obviously $\OPT(I,J) = \OPT \leq n$. Since we want to find an asymptotic FPTAS we may assume that $\OPT\geq 1$. We partition the interval $[1,n]$ into subintervals of size $\eps$. Because $\OPT \geq 1$ we know there exists a $\hat{t}$ such that $(1-\eps)\OPT \leq \hat{t}\eps \leq \OPT$.

For given $t$ let $\lambda(t)$ be the value of an optimal solution to LP~(\ref{lp:afptas}) and $\lambda^*(t)$ be the value of the solution to LP~(\ref{lp:afptas}) given by the price directive decomposition method. 

If $t\eps \leq \OPT$ then $\lambda(t) \leq 1$, since $\lambda'=1$ is the value of a solution, when $t\eps = \OPT$. In this case the price directive decomposition method finds a solution with $\lambda^*(t) \leq 1+\eps$. If the price directive decomposition method finds a solution with value $\lambda^*(t) > 1+\eps$ we know that there is no solution with value $\lambda(t') \leq 1$ for any $t'\geq t$. Hence by binary search we find a largest $t^*$ such that $\lambda(t^*) \leq 1+\eps$:

As argued there exists a $\hat{t}$ such that $(1-\eps)\OPT \leq \hat{t}\eps \leq \OPT$. Since for every $t' \leq \hat{t}$ a solution with value $\lambda(t') \leq 1+\eps$ can be found by the price directive decomposition method we find a $t^*\geq \hat{t}$. Thus $(1-\eps)\OPT \leq t^*\eps$.

Also, the solution vector $x^*$ corresponding to the solution with value $\lambda^*(t^*)$ may not be feasible, namely if $\lambda^*(t^*) > 1$. We can transform the solution vector $x^*$ into a solution $x'$ by multiplying each coordinate by the value $1-\eps$. It is easy to see that if $x^*$ is a solution for LP~(\ref{lp:afptas}), such that the left-hand side of each constraint has value $\lambda^*\leq 1+\eps$, then for the solution $x'$ in LP~(\ref{lp:afptas}) the left-hand side of each constraint has value at most $(1-\eps)\lambda^* \leq (1-\eps)(1+\eps) \leq 1-\eps^2 \leq 1$. Hence $x'$ is a feasible solution for LP~(\ref{lp:aptas}). As argued the objective value $\lambda^*(t^*)\geq (1-\eps)$ and hence the objective value $\lambda'$ for the scaled solution $x'$ is at least $(1-\eps)^2 \geq 1-2\eps$.

A solution $x'$ may have up to $\bigO{(k+1)(\eps^{-2}+  \ln (k+1))}$ coordinates, since there are so many calls to the block solver by the price directive decomposition method~\cite{JansenZhang:2002}. We can transform this solution $x'$ into a basic solution with at most $1+1/\eps^4$ fractional coordinates in order to improve the approximation ratio. This can be done by solving a homogeneous linear system of equalities, cf.~\cite{JansenSolis-oba:2003} for details.

We argue about the running time. An algorithm for a resource sharing problem with $M$ constraints given by Jansen and Zhang~\cite{JansenZhang:2002} finds a solution in $\bigO{M(\eps^{-2}+  \ln M)}$ iterations and has an overhead of $\bigO{M\ln\ln(M/\eps)}$ operations per step. In our case it is $M=k+1$. 
 
The dynamic program of Jansen and Solis-Oba in~\cite{JansenSolis-oba:2003}, which we use as a procedure for solving the block problem has a running time of $\bigO{n^2/\eps}$ per call. The overhead for scaling the price vector before we call this program is $O(k)$ and as we have $m$ calls to this program we need time $\bigO{m(k+n^2/\eps)}$ in order to solve the block problem. Note that, in particular, neither the running time of the block solver nor of the algorithm from~\cite{JansenZhang:2002} depends on the size of $|B_t|= \bigO{n^{1/\eps^4}}$.

We need an additional time of $\bigO{(k+1)(\eps^{-2}+  \ln (k+1)) \mathcal{M}(2+k)}$ for transforming the solution vector with $\bigO{(k+1)(\eps^{-2}+  \ln (k+1))}$ coordinates in an vector with $\bigO{1+k}$ coordinates, where $\mathcal{M}(2+k)$ is the running time for solving a homogeneous linear system of $2+1/\eps^4$ equations in $2+k$ variables. 
\end{proof}

\begin{proof}[of Theorem~\ref{thm:afptas}] The AFPTAS works identically as the APTAS, with the exception that we approximate LP~(\ref{lp:aptas}) as given by Lemma~\ref{lem:approxLPAPTAS}. We first argue about the approximation guarantee. We can proceed as in the proof of Theorem~\ref{thm:aptas}. After Inequality~(\ref{cor:s6}) we have to take into account that LP~(\ref{lp:aptas}) is only approximated. Then we can bound the additional loss of bins by the rounding procedure of the APTAS as done in Inequality~(\ref{cor:s7}). This gives then a feasible solution to the \prob{Variable-Sized Bin Covering} problem with an approximation guarantee at most $1+\eps$ minus a constant number of bins. The running time is dominated by approximating LP~(\ref{lp:aptas}) and hence given by Lemma~\ref{lem:approxLPAPTAS}.

\end{proof}

\pagebreak

\bibliography{bibliography}{}

\end{document}